\def\draft{0}
\documentclass[11pt]{article}

\usepackage{fullpage}
\usepackage{amsmath,amsfonts,amsthm,mathrsfs,mathpazo,xspace,hyperref,graphicx}
\usepackage{endnotes}
\usepackage{color}
\usepackage{bm}
\usepackage{amssymb,latexsym}
\usepackage{enumitem}
\usepackage{float}

\newtheorem{theorem}{Theorem}

\newtheorem{lemma}[theorem]{Lemma}

\newtheorem{fact}[theorem]{Fact}
\newtheorem{corollary}[theorem]{Corollary}
\newtheorem{definition}[theorem]{Definition}

\newtheorem{remark}[theorem]{Remark}

\newcommand{\beq}{\begin{eqnarray}}
\newcommand{\eeq}{\end{eqnarray}}

\newcommand{\ket}[1]{|#1\rangle}
\newcommand{\bra}[1]{\langle#1|}
\newcommand{\Tr}{\mbox{\rm Tr}}

\newcommand{\ot}{\otimes}

\newcommand{\C}{\ensuremath{\mathbb{C}}}

\DeclareMathOperator{\E}{\mathbf{E}}

\newcommand{\eps}{\varepsilon}

\DeclareMathOperator{\poly}{poly}

\ifnum\draft=1
\newcommand{\highlightname}[2]{{\textcolor{blue}{\footnotesize{\bf (#1:} {#2}{\bf ) }}}}
\newcommand{\mnote}[1]{{\highlightname{Matt}{#1}}}
\newcommand{\anote}[1]{\textcolor{red}{\footnotesize
{\textbf{(Anand:} #1\textbf{) }}}}
\else
\newcommand{\mnote}[1]{{}}
\newcommand{\anote}[1]{{}}
\fi

\newcommand\numberthis{\addtocounter{equation}{1}\tag{\theequation}}

\newcommand{\va}{\vec{a}}
\newcommand{\vb}{\vec{b}}
\newcommand{\vc}{\vec{c}}
\newcommand{\vd}{\vec{d}}
\newcommand{\ve}{\vec{e}}
\newcommand{\vf}{\vec{f}}
\newcommand{\vp}{\vec{p}}

\newcommand{\vr}{\vec{r}}
\newcommand{\vs}{\vec{s}}
\newcommand{\vt}{\vec{t}}
\newcommand{\vu}{\vec{u}}
\newcommand{\vv}{\vec{v}}

\newcommand{\HH}{\mathcal{H}}

\newcommand{\WA}{{W^A}}
\newcommand{\WAd}{{W^{A\dagger}}}
\newcommand{\WB}{{W^B}}
\newcommand{\WBd}{{W^{B\dagger}}}

\newcommand{\anck}[1]{\ket{\mathrm{ancilla}_{#1}}}
\newcommand{\proj}[1]{\ket{#1}\bra{#1}}
\newcommand{\EPR}{\ket{\mathrm{EPR}}}
\newcommand{\pEPR}{\proj{\mathrm{EPR}}}
\newcommand{\tA}{\tilde{A}}
\newcommand{\tB}{\tilde{B}}

\newcommand{\thmref}[1]{Theorem~\ref{thm:#1}}
\newcommand{\lemref}[1]{Lemma~\ref{lem:#1}}

\newcommand{\myeqref}[1]{Equation~\eqref{eq:#1}}
\newcommand{\figref}[1]{Figure~\ref{fig:#1}}
\newcommand{\secref}[1]{Section~\ref{sec:#1}}
\renewcommand{\vec}[1]{\mathbf{#1}}
\newcommand{\cd}{\cdot}
\begin{document}

\title{The Parallel-Repeated Magic Square Game is Rigid}
\author{Matthew Coudron\thanks{Massachusetts Institute of Technology. Email\texttt{\href{mailto:mcoudron@mit.edu}{\color{black}mcoudron@mit.edu}}. }
  \qquad  Anand Natarajan \thanks{Center for Theoretical Physics, MIT. {\tt anandn@mit.edu}}
}

\maketitle

\begin{abstract}
  We show that the $n$-round parallel repetition of the Magic Square
  game of Mermin and Peres is rigid, in the sense that for any entangled strategy
  succeeding with probability $1 -\varepsilon$, the players' shared state is
  $O(\mathrm{poly}(n\varepsilon))$-close to $2n$ EPR pairs under a local
  isometry.  Furthermore, we show that, under local isometry, the
  players' measurements in said entangled strategy must be
  $O(\mathrm{poly}(n\varepsilon))$-close to the ``ideal" strategy when acting on the
  shared state.  
\end{abstract}

\section{Introduction}
Nonlocal games have long been a fundamental topic in quantum
information, starting from Bell's pioneering work in the 1960s. In the
langauge of games, Bell~\cite{Bell64} showed that for a certain two-player nonlocal
game, two players sharing a single EPR pair between them can win with
substantially higher probability than they could by following the best
classical strategy. In Bell's original game, the messages between the
players and the referee were real numbers, but soon afteward, Clauser,
Horne, Shimony, and Holt~\cite{CHSH69} discovered a game (called the CHSH game) with similar properties,
but with messages consisting of just one bit. The CHSH game can be
viewed as a \emph{test} for the ``quantumness'' of a system, with good
\emph{soundness}: that is, the probability of a non-quantum system
fooling the test is at most $3/4$. However, the test lacks the
property of so-called \emph{perfect completeness}: as shown by
Tsirelson~\cite{Tsirelson80}, even the optimal
quantum strategy succeeds with probability at most $(2 + \sqrt{2})/4
\approx 0.854$. To remedy this drawback, Mermin~\cite{Mermin90} and
independently Peres~\cite{Peres90}
independently introduced the \emph{Magic Square game}: a two-player
game with two-bit inputs and outputs, and for which the best classical
strategy succeeds with probability $8/9$, but there exists a quantum
strategy using only two shared EPR pairs succeeding with probability
$1$. 

Later, Mayers and Yao~\cite{MY98} realized that the CHSH game could be used not only
to test for ``quantumness,'' but to test for a \emph{specific} quantum
state: namely, the EPR pair. Such a test is often called a ``self-test.'' Mayers and Yao showed that in any optimal quantum
strategy for CHSH, the players' shared state is equivalent under a
local isometry\footnote{Since either player could apply a
local unitary to their half of the state and their measurements,
without affecting their winning probability, equivalence under local
isometry is the best one could hope for.} to an EPR
pair.  This result
was not \emph{robust} in that required the CHSH correlations to hold
\emph{exactly}: however, the subsequent work of McKague, Yang, and
Scarani~\cite{MYS12} was able to achieve a robust self-test based on
CHSH for a single EPR pair. That is, they showed that for any strategy
that wins CHSH with probability $\geq p_{\max} - \eps$, there exists
an isometry $V$ mapping the players' state $\ket{\psi}$ to a state
$\ket{\phi}$ which is
$O(\sqrt{\eps})$-close to the EPR pair state in 2-norm. Moreover, they
showed that the \emph{measurements} applied by the players must also
be close to the measurements used in the ideal strategy, as measured
in a state-dependent distance: for instance, if $X$ is the operator
applied by player 1 when asked to measure a Pauli $X$, then under the
same isometry $V$, $\| V(X\ket{\psi}) - \sigma_X \ket{\phi} \| \leq
O(\sqrt{\eps})$, where $\sigma_X$ is the Pauli $X$-matrix. Such a
result is called a \emph{rigidity} result, because it shows that any
strategy that is close to optimal must have the same structure as the
ideal strategy. We refer to the bound that appears in the right-hand
side of the norm inequalities (here $\sqrt{\eps}$) as the
\emph{robustness} of the test. More recently, Wu et al.~\cite{WBMS16}
showed rigidity for Mermin and Peres's Magic Square game,
demonstrating that it serves as robust self-test for a single EPR pair. 

In recent years, self-testing has found applications to quantum cryptography (QKD,
device independent QKD, and randomness expansion), as well as to
multiprover quantum interactive proof systems (the complexity class MIP*)
\cite{RUV13}. However, these applications all rely on testing
\emph{multi-qubit} states, whereas known robust self-testing
results are directly applicable only to states of a few
qubits. A natural strategy to obtain a multi-qubit test is to
\emph{repeat} the single-qubit tests, either in series (i.e. over many
rounds) or in parallel (i.e. in one round)---for instance, the work of
Reichardt, Unger, and Vazirani~\cite{RUV13} uses a serially repeated
CHSH test, and McKague~\cite{McK15} gives a parallel self-test based
on CHSH. The lack of perfect completeness considerably complicates the
analysis of these tests, since one cannot demand that the players win
\emph{every} repetition of the test---rather, one has to check whether
the fraction of successful repetitions is above a certain threshold. 

In this paper, we circumvent these issues by studying the $n$-round parallel
repetition of the Magic Square game.  We achieve a proof
of rigidity, showing that if the players win with probability $1 -
\eps$, their state is $O(\poly(n\eps))$-close to $2n$ EPR pairs, under
a local isometry.  This is
an exponential improvement in error dependence over the strictly parallel
self-testing result of \cite{McK15}, which has error depedence
$O(\exp(n) \poly(\eps))$  \footnote{Note that, by repeating the test in section 4 of \cite{McK15} a polynomial number of times, one can achieve a self-test for $n$ EPR pairs with polynomial error dependence.  However, the test given in section 4 is not a strictly parallel test, and does not robustly certify $n$-qubit measurement operators, as our result does.  }, and is the previous best known result for
rigidity of strictly parallel repeated non-local games (McKague's result is stated
for the parallel repeated CHSH game with a threshold test, rather than
the parallel repeated Magic Square game). We note
that McKague's result
has $O(\log(n))$-bit questions, whereas our game has $O(n)$-bit questions and answers, but additionally robustly certifies all $n$-qubit measurement operators.  This means that our result is a strictly parallel test, that can be used to "force'' untrusted provers to apply all $n$-qubit Pauli operators faithfully (in expectation), which is a new feature that we believe will be valuable in the context of complexity applications.  

As a fundamental building
block for our result, we make use of the rigidity of a single
round of the Magic Square game, which was established in
\cite{WBMS16}.  A key observation of our work is that, by leveraging a ``global consistency check" which occurs naturally within the parallel repeated Magic Square game, we can establish approximate commutation between the different copies (or ``rounds") of the game in the parallel repeated test.  This then allows us to extend the single round analysis of \cite{WBMS16}, to a full $n$-round set of approximate anti-commutation relations for the provers measurements, which is expressed in Theorem \ref{thm:approx_phase}.  A second important technical tool in our proof is a
theorem (\thmref{iso}) which, given operators on the players' state
that approximately satisfy the algebraic relations of single-qubit
Pauli matrices, constructs an isometry that maps the players' ``approximate Paulis'' close to exact Pauli operators acting on a $2n$-qubit
space.  The proof of \thmref{iso} relies on an isometry inspired by the works of McKague~\cite{McK10,McK16}, but is designed to take the guarantees produced by Theorem \ref{thm:approx_phase} and conclude closeness of the players ``approximate Paulis'' to exact Pauli operators in expectation, where all $2n$-qubit Pauli operators are handled simultaneously, with polynomial error dependence.

Very recently, we became aware of two independent works achieving
related results in this area.  The first is an unpublished paper of
Chao, Reichardt, Sutherland, and Vidick~\cite{CRSV16}, which proves a theorem similar to our \thmref{iso}.  The second is a paper by
Coladangelo~\cite{Coladangelo16}, which proves a
self-testing result for the parallel repeated Magic Square game that is similar our own,
albeit with slightly different polynomial factors. Furthermore, the robustness analysis of the results in \cite{Coladangelo16} makes use of the same key theorem of~\cite{CRSV16}, which is, in turn, similar to our own \thmref{iso}. The
theorem of~\cite{CRSV16} (and consequently the robustness result of
Coladangelo) achieve a robustness of $n^{3/2}\sqrt{\eps}$ for 
for all single-qubit operators (i.e., to achieve constant robustness,
$\eps$ must scale as $1/n^3$).  On the other hand, our \thmref{iso}
achieves a robustness of $n \eps^{1/4}$ (i.e. $\eps \sim 1/n^4$), but for operators acting on \emph{all}
$2n$ qubits simultaneously.  It is natural to ask whether one
can prove a single result which combines the strengths of these two different error dependencies.  We expect that this is possible, but leave it for future work.

\section{Preliminaries}
We use the standard quantum formalism of states and measurements. An
\emph{observable} is a Hermitian operator whose eigenvalues are $\pm
1$, and encodes a two-outcome projective measurement (the POVM elements of the two outcomes
are the projections on to the $+1$ and $-1$ eigenspaces).
Throughout this paper, we make use of the Pauli matrices. These are $2
\times 2$ Hermitian matrices defined by
\[ \sigma_X := \begin{pmatrix} 0 & 1 \\ 1 & 0 \end{pmatrix}, \sigma_Z
:= \begin{pmatrix} 1 & 0 \\ 0 & -1 \end{pmatrix}, \sigma_Y
:= \begin{pmatrix} 0 & -i \\ i & 0 \end{pmatrix}. \]
They satisfy the anticommutation relation
\[ XZ = -ZX. \]

\section{The Magic Square game}
In this section we introduce the nonlocal game analyzed in this work:
the $n$-round parallel repeated Magic Square game.  We also introduce
notation to describe entangled strategies for the game and state some
simple properties they satisfy.

The parallel repeated Magic Square game is played between
players (which we will refer to as Alice and Bob), and a
verifier. First, let us define the single-round Magic Square game,
originally introduced by Mermin~\cite{Mermin90} and
Peres~\cite{Peres90}. The rules of the game are
described in Fig.~\ref{fig:magic}.
\begin{figure}[H]
  \rule[1ex]{16.5cm}{0.5pt}\\
  The magic square game is a one-round, two-player game, played as follows
  \begin{enumerate}
  \item The verifier sends Alice a question $r \in \{0,1,2\}$ and Bob a
    question $c \in \{0,1,2\}$.
  \item Alice sends the verifier a response $(a_0, a_1) \in \{0,1\}^2$,
    and Bob sends a response $(b_0, b_1) \in \{0, 1\}^2$.
  \item Let $a_2 := a_0 \oplus a_1$ and $b_2 := 1 \oplus b_0 \oplus b_1$. Then
    Alice and Bob win the game if $a_c = b_r$ and lose otherwise.
  \end{enumerate}
  \rule[1ex]{16.5cm}{0.5pt}
  \caption{The magic square game}
  \label{fig:magic}
\end{figure}
Any entangled strategy for this game is described by a shared quantum
state $\ket{\psi}_{AB}$ and projectors $P_{r}^{a_0, a_1}$ for Alice
and $Q_{c}^{b_0, b_1}$ for Bob. It can be seen that the game can be
won with certainty for the following strategy:
\begin{align*}
  \ket{\psi} &= \frac{1}{2} \sum_{i, j \in \{0, 1\}} \ket{ij}_A \ot
               \ket{ij}_B \\
  P_{0}^{a_0, a_1} &= \frac{1}{4}(I + (-1)^{a_0} Z)_{A1}
                     \ot (I + (-1)^{a_1} Z)_{A2} \ot I_{B} \\
  P_{1}^{a_0, a_1} &= \frac{1}{4}  (I + (-1)^{a_1} X)_{A1} \ot (I +
                     (-1)^{a_0} X)_{A2}
                     \ot I_{B} \\
  Q_{0}^{b_0, b_1} &= \frac{1}{4}I_A \ot (I + (-1)^{b_0} Z)_{B1} \ot
                     (I + (-1)^{b_1} X)_{B2} \\
  Q_{1}^{b_0, b_1} &= \frac{1}{4} I_A \ot (I + (-1)^{b_1} X)_{B1} \ot
                     (I + (-1)^{b_0} Z)_{B2} 
\end{align*}
This strategy is represented pictorially in
Fig.~\ref{fig:1roundideal}, where each row contains a set of
simultaneously-measurable observables that give Alice's answers, and
likewise each column for Bob.

\begin{figure}[H]
  \centering
  \begin{tabular}{c | c | c}
    ZI & IZ & ZZ \\ \hline
    IX & XI & XX \\ \hline
    -ZX & -XZ & YY
  \end{tabular}
  \caption{The ideal strategy for a single round of magic square. Alice
    and Bob share the state $\ket{\text{EPR}}^{\ot 2}$.}
  \label{fig:1roundideal}
\end{figure}

The game we study in this paper is the $n$-fold parallel repetition of
the above game.
\begin{definition}
  \label{def:parallel_magic}
  The $n$-fold parallel repeated Magic Square game is a game with two
  players, Alice and Bob, and one verifier. The player sends Alice a
  vector $\vec{r} \in \{0, 1,2\}^{n}$ and Bob a vector $\vec{c} \in
  \{0, 1, 2\}^{n}$, where each coordinate of $\vr$ and $\vc$ is chosen
  uniformly at random. Alice responds with two $n$-bit strings
  $\vec{a}_0, \vec{a}_1$, and Bob with two $n$-bit strings $\vec{b}_0,
  \vec{b}_1$. The players win if for every $k \in [n]$, the $k$th
  components of Alice and Bob's answers $a_{0,k}, a_{1,k}, b_{0,k},
  b_{1,k}$ satisfy the win conditions of the Magic Square game with
  input $r_k$ and $c_k$.
\end{definition}
Throughout this paper we will refer to the non-local entangled
strategy applied by the players according to the following
definitions:

\begin{definition}
  Let $\{ P^{\vec a_0, \vec a_1 }_{\vec r} \}_{\vec a_0, \vec a_1}$ denote the set of orthogonal projectors describing Alice's measurement when she receives input $\vec r$.

  Likewise, let $\{ Q^{\vec b_0, \vec b_1 }_{\vec c} \}_{\vec b_0, \vec b_1}$ denote the set of orthogonal projectors describing Bob's measurement when he receives input $\vec c$. 
\end{definition}

\begin{definition}
  Define $\vec a_{2} \equiv \vec a_{0} + \vec a_{1} \pmod 2$ an $\vec b_{2} \equiv \vec b_{0} + \vec b_{1} + \vec 1 \pmod 2$.

\end{definition}

\begin{definition}
  Define the column-$\vec c$ output observables for Alice as $A_{\vec r, \vec p}^{\vec c} \equiv \sum_{\vec a_0, \vec a_1}(-1)^{\vec a_{\vec c} \cdot \vec p}P^{\vec a_0, \vec a_1 }_{\vec r}$.  

  Where $a_{\vec c}$ is defined to be the $n$ dimensional vector whose $i^{th}$ component is defined by $(a_{\vec c})_{i} \equiv (\vec{a_{c_i}})_{i}$.

  Similarly, define the row-$\vec r$ observables for Bob as $B_{\vec c, \vec q}^{\vec r} \equiv \sum_{\vec b_0, \vec b_1}(-1)^{ \vec{b}_{\vec{r}} \cdot \vec q}Q^{\vec b_0, \vec b_1 }_{\vec c}$.

\end{definition} 
\begin{remark}
  By definition, it follows that $A_{\vec r, \vec p}^{\vec c} =
  A_{\vec r, \vec{p} }^{\vec{c'}}$ if $\vec{c}$ and $\vec{c'}$ differ only
  on rounds where the coordinate of $\vec{p}$ is $0$, and likewise for
  $B$ and $\vec{r}$.
  \label{rem:ignore_output}
\end{remark}
The win conditions for magic square:
\begin{fact} \label{fact:magicwin}
  Suppose Alice and Bob win the magic square game with probability $\geq
  1 - \eps$. Then it holds that 
  \begin{equation}
    \forall \vec{p}, \quad \E_{\vec{r}, \vec{c}} \bra{\psi} A_{\vec r, \vec p}^{\vec c} B_{\vec
      c, \vec p}^{\vec r} \ket{\psi} \geq 1 - \eps.
    \label{eq:magicwin}
  \end{equation}
\end{fact}

In Remark~\ref{rem:ignore_output}, we noted that we can freely change
the output column for Alice (resp. row for Bob) on the ``ignored''
rounds. In the following lemma, we show that we can also change the
\emph{input} row (resp. column), up to an $O(\eps)$ error, provided
that the strategy is $\eps$ close to optimal.

\begin{lemma} \label{lem:moveatoa}
  Suppose Alice and Bob have an $\eps$-optimal strategy. Then, $\forall i, r, \vec c$,
  \[\left |  1 -  \E_{\vec r, \vec{r'}:  r'_{i} = r_i = r} \bra{\psi} A^{\vec c}_{\vec r,  \vec{e}_i} \cdot A_{\vec r' , \vec{e}_i}^{\vec c} \ket{\psi} \right |   \leq 36 \eps \]

  \anote{Not too sure about the distribution, is it correct?}
\mnote{We think this should work for the uniform distribution, not sure if works for the distributions that just have uniform marginals.  Should check.}
\end{lemma}
\begin{proof}
  To start we define an extended state $\ket{\sigma} \equiv \ket{\psi} \ot \frac{1}{\sqrt{3^{-(n-1)}}}\sum_{\vec r_{-i}} \ket{\vec r_{-i}} \ot \frac{1}{\sqrt{3^{-(n-1)}}}\sum_{\vec r'_{-i}} \ket{\vec r'_{-i}} \ot \frac{1}{\sqrt{3^{-(n-1)}}}\sum_{\vec s_{-i}} \ket{\vec s_{-i}} $ as well as extended operators:
  \begin{align*}
  & T \equiv \sum_{\vec r_{-i}} A^{\vec c}_{\vec r,  \vec{e}_i}  \ot  \proj{\vec r_{-i}} \ot I \ot  I  = \sum_{\vec r_{-i}} \sum_{\vec s_{-i}} A^{c_i \cup \vec s_{-i}}_{\vec r, \vec{e}_i}  \ot  \proj{\vec r_{-i}} \ot I \ot  \proj{\vec s_{-i}}\\
  & \intertext{Note that, by Remark \ref{rem:ignore_output}, these two definitions are equivalent because $A^{\vec c}_{\vec r, \cdot \vec{e}_i} $ is identically equal to $A^{c_i \cup \vec s_{-i}}_{\vec r, \cdot \vec{e}_i}$ by definition, regardless of the value of $\vec s_{-i}$.  Further define }
  & T' \equiv \sum_{\vec r'_{-i}} A^{\vec c}_{\vec r',  \vec{e}_i} \ot I  \ot  \proj{\vec r'_{-i}} \ot I = \sum_{\vec r'_{-i}} \sum_{\vec s_{-i}} A^{c_i \cup \vec s_{-i}}_{\vec r',  \vec{e}_i} \ot I  \ot  \proj{\vec r'_{-i}} \ot  \proj{\vec s_{-i}}\\
  &\intertext{and}
  & S \equiv \sum_{\vec r_{-i}} \sum_{\vec s_{-i}} B^{r_i \cup \vec r_{-i}}_{\vec c_i \cup s_{-i},  \vec{e}_i} \ot  \proj{\vec r_{-i}}  \ot I \ot  \proj{\vec s_{-i}} = \sum_{\vec r_{-i}} \sum_{\vec s_{-i}} B^{r_i \cup \vec r_{-i}}_{\vec c_i \cup s_{-i},  \vec{e}_i} \ot  \proj{\vec r_{-i}}  \ot \sum_{\vec r'_{-i}} \proj{\vec r'_{-i}} \ot  \proj{\vec s_{-i}} \\
  &= \sum_{\vec r'_{-i}} \sum_{\vec s_{-i}} B^{r'_i \cup \vec r'_{-i}}_{\vec c_i \cup s_{-i},  \vec{e}_i} \ot \sum_{\vec r_{-i}} \proj{\vec r_{-i}}  \ot  \proj{\vec r'_{-i}} \ot \proj{\vec s_{-i}} = \sum_{\vec r'_{-i}} \sum_{\vec s_{-i}} B^{r_i \cup \vec r'_{-i}}_{\vec c_i \cup s_{-i},  \vec{e}_i} \ot I  \ot  \proj{\vec r'_{-i}} \ot  \proj{\vec s_{-i}}
  \end{align*}
  
  Where, to conclude equivalence of the different versions of the last definition, we are using Remark \ref{rem:ignore_output} as well as the fact that $r_i = r'_i = r$, some fixed value.

  Now, note that:
  
  \begin{align*}
  &\bra{\sigma} T \cdot S \ket{\sigma} = \left (\bra{\psi} \ot \frac{1}{\sqrt{3^{-(n-1)}}}\sum_{\vec r_{-i}} \bra{\vec r_{-i}} \ot \frac{1}{\sqrt{3^{-(n-1)}}}\sum_{\vec r'_{-i}} \bra{\vec r'_{-i}} \ot \frac{1}{\sqrt{3^{-(n-1)}}}\sum_{\vec s_{-i}} \bra{\vec s_{-i}} \right )   \\
  & \times \left ( \sum_{\vec r_{-i}} \sum_{\vec s_{-i}} A^{c_i \cup \vec s_{-i}}_{\vec r, \vec{e}_i}  \ot  \proj{\vec r_{-i}} \ot I \ot  \proj{\vec s_{-i}} \right ) \left (\sum_{\vec r_{-i}} \sum_{\vec s_{-i}} B^{r_i \cup \vec r_{-i}}_{\vec c_i \cup s_{-i},  \vec{e}_i} \ot  \proj{\vec r_{-i}}  \ot I \ot  \proj{\vec s_{-i}} \right ) \\
  & \times \left (\ket{\psi} \ot \frac{1}{\sqrt{3^{-(n-1)}}}\sum_{\vec r_{-i}} \ket{\vec r_{-i}} \ot \frac{1}{\sqrt{3^{-(n-1)}}}\sum_{\vec r'_{-i}} \ket{\vec r'_{-i}} \ot \frac{1}{\sqrt{3^{-(n-1)}}}\sum_{\vec s_{-i}} \ket{\vec s_{-i}} \right ) \\
  & = \frac{1}{3^{-2(n-1)}} \sum_{\vec r_{-i}, \vec s_{-i}} \bra{\psi} A^{c_i \cup \vec s_{-i}}_{\vec r, \vec{e}_i} B^{r_i \cup \vec r_{-i}}_{\vec c_i \cup s_{-i},  \vec{e}_i} \ket{\psi} \cdot \left ( \frac{1}{\sqrt{3^{-(n-1)}}}\sum_{\vec r'_{-i}} \bra{\vec r'_{-i}}\right ) \left( \frac{1}{\sqrt{3^{-(n-1)}}}\sum_{\vec r'_{-i}} \ket{\vec r'_{-i}}\right) \\
  & = \frac{1}{3^{-2(n-1)}} \sum_{\vec r_{-i}, \vec s_{-i}} \bra{\psi} A^{c_i \cup \vec s_{-i}}_{\vec r, \vec{e}_i} B^{r_i \cup \vec r_{-i}}_{\vec c_i \cup s_{-i},  \vec{e}_i} \ket{\psi}  = \E_{\vec r_{-i}, \vec s_{-i}}  \bra{\psi} A^{c_i \cup \vec s_{-i}}_{\vec r, \vec{e}_i} B^{r_i \cup \vec r_{-i}}_{\vec c_i \cup s_{-i},  \vec{e}_i} \ket{\psi} \geq 1 - 9 \eps
  \end{align*}
  Where the last line follows by Fact \ref{fact:magicwin}.  Similarly,
    \begin{align*}
    &\bra{\sigma} T' \cdot S \ket{\sigma} = \left (\bra{\psi} \ot \frac{1}{\sqrt{3^{-(n-1)}}}\sum_{\vec r_{-i}} \bra{\vec r_{-i}} \ot \frac{1}{\sqrt{3^{-(n-1)}}}\sum_{\vec r'_{-i}} \bra{\vec r'_{-i}} \ot \frac{1}{\sqrt{3^{-(n-1)}}}\sum_{\vec s_{-i}} \bra{\vec s_{-i}} \right )   \\
    & \times \left (\sum_{\vec r'_{-i}} \sum_{\vec s_{-i}} A^{c_i \cup \vec s_{-i}}_{\vec r',  \vec{e}_i} \ot I  \ot  \proj{\vec r'_{-i}} \ot  \proj{\vec s_{-i}}\right ) \left (\sum_{\vec r'_{-i}} \sum_{\vec s_{-i}} B^{r_i \cup \vec r'_{-i}}_{\vec c_i \cup s_{-i},  \vec{e}_i} \ot I  \ot  \proj{\vec r'_{-i}} \ot  \proj{\vec s_{-i}}\right ) \\
    & \times \left (\ket{\psi} \ot \frac{1}{\sqrt{3^{-(n-1)}}}\sum_{\vec r_{-i}} \ket{\vec r_{-i}} \ot \frac{1}{\sqrt{3^{-(n-1)}}}\sum_{\vec r'_{-i}} \ket{\vec r'_{-i}} \ot \frac{1}{\sqrt{3^{-(n-1)}}}\sum_{\vec s_{-i}} \ket{\vec s_{-i}} \right ) \\
    & = \frac{1}{3^{-2(n-1)}} \sum_{\vec r'_{-i}, \vec s_{-i}} \bra{\psi} A^{c_i \cup \vec s_{-i}}_{\vec r', \vec{e}_i} B^{r_i \cup \vec r'_{-i}}_{\vec c_i \cup s_{-i},  \vec{e}_i} \ket{\psi} \cdot \left ( \frac{1}{\sqrt{3^{-(n-1)}}}\sum_{\vec r_{-i}} \bra{\vec r_{-i}}\right ) \left( \frac{1}{\sqrt{3^{-(n-1)}}}\sum_{\vec r_{-i}} \ket{\vec r_{-i}}\right) \\
    & = \frac{1}{3^{-2(n-1)}} \sum_{\vec r'_{-i}, \vec s_{-i}} \bra{\psi} A^{c_i \cup \vec s_{-i}}_{\vec r', \vec{e}_i} B^{r_i \cup \vec r'_{-i}}_{\vec c_i \cup s_{-i},  \vec{e}_i} \ket{\psi}  = \E_{\vec r'_{-i}, \vec s_{-i}}  \bra{\psi} A^{c_i \cup \vec s_{-i}}_{\vec r', \vec{e}_i} B^{r_i \cup \vec r'_{-i}}_{\vec c_i \cup s_{-i},  \vec{e}_i} \ket{\psi} \geq 1 - 9 \eps
    \end{align*}

  Where the last line again follows by Fact \ref{fact:magicwin}.  It follows by Lemma \ref{lem:maybesaveeps}, that 
  
  \[\bra{\sigma} T \cdot T' \ket{\sigma} \geq 1-  36 \eps \]
  
  Noting that  
  
  \begin{align*}
  &\bra{\sigma} T \cdot T' \ket{\sigma} = \left (\bra{\psi} \ot \frac{1}{\sqrt{3^{-(n-1)}}}\sum_{\vec r_{-i}} \bra{\vec r_{-i}} \ot \frac{1}{\sqrt{3^{-(n-1)}}}\sum_{\vec r'_{-i}} \bra{\vec r'_{-i}} \ot \frac{1}{\sqrt{3^{-(n-1)}}}\sum_{\vec s_{-i}} \bra{\vec s_{-i}} \right )   \\
  & \times \left (\sum_{\vec r_{-i}} A^{\vec c}_{\vec r,  \vec{e}_i}  \ot  \proj{\vec r_{-i}} \ot I \ot  I \right  ) \left (\sum_{\vec r'_{-i}} A^{\vec c}_{\vec r',  \vec{e}_i} \ot I  \ot  \proj{\vec r'_{-i}} \ot I\right ) \\
  & \times \left (\ket{\psi} \ot \frac{1}{\sqrt{3^{-(n-1)}}}\sum_{\vec r_{-i}} \ket{\vec r_{-i}} \ot \frac{1}{\sqrt{3^{-(n-1)}}}\sum_{\vec r'_{-i}} \ket{\vec r'_{-i}} \ot \frac{1}{\sqrt{3^{-(n-1)}}}\sum_{\vec s_{-i}} \ket{\vec s_{-i}} \right ) \\
  & = \frac{1}{3^{-2(n-1)}} \sum_{\vec r_{-i}, \vec r'_{-i}} \bra{\psi} A^{\vec c}_{\vec r,  \vec{e}_i} A^{\vec c}_{\vec r',  \vec{e}_i} \ket{\psi}  = \E_{\vec r_{-i}, \vec r'_{-i}:  :  r'_{i} = r_i = r}  \bra{\psi} A^{\vec c}_{\vec r,  \vec{e}_i} A^{\vec c}_{\vec r',  \vec{e}_i} \ket{\psi}
  \end{align*}

  So, we have,
  
  \begin{align*}
  &  \left |  1 -  \E_{\vec r, \vec{r'}:  r'_{i} = r_i = r} \bra{\psi} A^{\vec c}_{\vec r,  \vec{e}_i} \cdot A_{\vec r' , \vec{e}_i}^{\vec c} \ket{\psi} \right |  =  \left | 1 - \bra{\sigma} T \cdot T' \ket{\sigma} \right | \leq 36 \eps
  \end{align*}
\end{proof}
\section{Results}
In this section, we state and prove our technical results on the
structure of strategies for the parallel repeated Magic Square
game. We first give an overview of the proof and then fill in the
technical details.
\subsection{Overview}
Our result has two main technical components. The first is a theorem
that, given a near-optimal strategy, shows how to construct
observables on each players' Hilbert space that approximately satisfy a set of
pairwise commutation and anticommutation relations.  
\begin{theorem}
  Suppose that two players Alice and Bob have an entangled strategy
  for the $n$-round parallel repeated Magic Square game, which wins
  with probability at least $1-\eps$.  Then, if we adjoin an ancilla
  register to Alice's space in the appropriate state
  $\ket{\text{ancilla}}_{A}$ (and similarly for Bob in the appropriate state $\ket{\text{ancilla}}_{B}$), there exist observables
  $\tA^{c}_{r, k}$ indexed by $r, c \in \{0,1,2\}$ and $k \in [n]$
  acting on Alice's space such that 
  \begin{equation}
    \begin{aligned}
      &\forall k, r, c, r', c', \qquad d_{\psi'}(\tA^c_{r,k} \tA^{c'}_{r',k},
      (-1)^{f(r,r',c,c')} \tA^{c'}_{r',k} \tA^c_{r,k}) &\leq O(\sqrt{\eps})
      \\
      &\forall k \neq k', r, c, r', c', \qquad d_{\psi'}(\tA^c_{r,k}
      \tA^{c'}_{r',k'}, \tA^{c'}_{r',k'} \tA^{c}_{r,k}) &\leq O(\sqrt{\eps}).
    \end{aligned}
    \label{eq:magic_anticom_alice}
  \end{equation}
  \anote{Fix alignment above}
  where $\ket{\psi'} =\ket{\psi} \ot \ket{\text{ancilla}}_{A} \ot \ket{\text{ancilla}}_{B}  $ denotes the state together with the ancilla
  registers, and $f(r,r', c, c') = 1$ if $r \neq r'$ and $c \neq c'$,
  and $0$ otherwise.

  Likewise, there exist
  observables $\tB^{c}_{r, k}$ on Bob's
  space such that
  \begin{equation}
    \begin{aligned}
      &\forall k, r, c, r', c', \qquad d_{\psi'}(\tB^r_{c,k} \tB^{r'}_{c',k},
      (-1)^{f(r,r',c,c')} \tB^{r'}_{c',k} \tB^r_{c,k}) &\leq O(\sqrt{\eps})
      \\
      &\forall k \neq k', r, c, r', c', \qquad d_{\psi'}(\tB^r_{c,k}
      \tB^{r'}_{c',k'}, \tB^{r'}_{c',k'} \tB^{r}_{c,k}) &\leq O(\sqrt{\eps}).
    \end{aligned}
    \label{eq:magic_anticom_bob}
  \end{equation}

  Moreover, the following consistency relations hold in expectation:
  \begin{align}
    \forall \vc, \vp, \quad \E_{\vr} d_{\psi'}(A^{\vc}_{\vr, \vp} \ot
    I_{\text{ancilla}}, \prod_{k=1}^{n} (\tA^{c_k}_{r_k, k})^{p_k})^2 &\leq
                                                  O(n \sqrt{\eps})  \label{eq:consistency_alice}
    \\
    \forall \vr, \vp, \quad \E_{\vc} d_{\psi'}(B^{\vr}_{\vc, \vp} \ot
    I_{\text{ancilla}}, \prod_{k=1}^{n} (\tB^{r_k}_{c_k, k})^{p_k})^2 &\leq
                                                  O(n \sqrt{\eps})  \label{eq:consistency_bob}
  \end{align}
  \label{thm:approx_phase}
\end{theorem}
\begin{proof}[Proof of \thmref{approx_phase}]
  The single-round phase relations in
  Equations~\eqref{eq:magic_anticom_alice}
  and~\eqref{eq:magic_anticom_bob} follow from
  \lemref{single_round_phase}. The commutation relations between
  rounds follow from \lemref{com_between_rounds}. The consistency
  relations (Equations~\eqref{eq:consistency_alice}
  and~\eqref{eq:consistency_bob}) follow from
  \lemref{consistency_alice}.
  \anote{Make sure to add the analogous statements for Bob to the lemmas!}
\end{proof}
Having constructed these observables, we use them to build an isometry
that ``extracts'' a $2n$-qubit state out of the shared state of Alice
and Bob. This isometry is \emph{local}: it does not create any
entanglement between Alice and Bob. Moreover, it maps the
measurements in the players' strategy to $2n$-qubit measurements that
are close to the ideal strategy.
\begin{theorem}
  \label{thm:iso}
  Suppose that two players \mnote{Alice and Bob} share an entangled
  state in a Hilbert space $\HH$ and operators $\tA^{c}_{r,k}, \tB^{r}_{c,k}$ satisfying
  Equations~\eqref{eq:magic_anticom_alice}
  and~\eqref{eq:magic_anticom_bob}. Then there exists an isometry
  $V: \HH \to \HH \ot \C^{2n} \ot \C^{2n} \ot \C^{2n} \ot \C^{2n}$,
  and for every $\vs, \vt \in
  \{0,1\}^{2n}$, there exists an operator $\WA_{\vs, \vt}$ on Alice's 
  space, and for every $\vu, \vv \in \{0,1\}^{2n}$ there exists an
  operator $\WB_{\vu, \vv}$ on Bob's space, such that
  \begin{align}
    \forall \va, \vb, \vc, \vd, \qquad \left| \bra{\phi} \sigma_X^A(\vs) \sigma_Z^A(\vt)  \sigma_X^B(\vu)
                  \sigma_Z^B(\vv) \ket{\phi} -\bra{\psi} \WA_{\vs,\vt}
    \WB_{\vu,\vv} \ket{\psi} \right| \leq O(n^2 \sqrt{\eps}),
  \end{align}
  where $\ket{\phi} = V(\ket{\psi})$, $\sigma_X^A, \sigma_Z^A$ are
  Pauli operators acting on the second output register of $V$, and
  $\sigma_X^B, \sigma_Z^B$ are Pauli operators acting on the fourth
  output register of $V$.
\end{theorem}
The proof of this theorem is deferred to \secref{isometry}. As a
corollary, we show that the output state of the isometry has high
overlap with the state $\EPR^{\ot 2n}$ consisting of $2n$ EPR pairs
shared between Alice and Bob.
\begin{corollary}
  Suppose that two players \mnote{Alice and Bob} have an entangled
  strategy for the $n$-round parallel repeated Magic Square game,
  which wins with probability at least $1-\eps$.  Then, letting
  $\ket{\phi} = V(\ket{\psi})$ as in \thmref{iso},
  \[ \bra{\phi} \pEPR^{\ot 2n} \ot I_{\text{junk}} \ket{\phi} \geq 1 -
  O(n^2 \sqrt{\eps}), \]
  whwere the identity operator $I_{\text{juk}}$ acts on the first, third, and fifth
  register of the isometry output.
\end{corollary}
\begin{proof}
This follows from \lemref{parallel_ideal_test} and \lemref{honest_measurement}.
\end{proof}
\subsection{Single-round observables}
\label{sec:single_round_observables}
\begin{definition}
  Let $k \in [n]$ be the index of a round, and denote the single round
  observables associated with that round by $A^{c}_{r,k} :=
  A^{\vc}_{\vr, \vec{e}_k}$ and $B^{r}_{c,k} := B^{\vr}_{\vc, \vec{e}_k}$, where $\vc$ and $\vr$ are any vectors
  whose $k$th coordinates are $r$ and $c$ respectively, and
  $\vec{e}_k$ is the vector with a $1$ in the $k$th position and $0$s
  elsewhere.  
\end{definition}

\begin{definition}
  For each round $k$, define the state $\ket{\mathrm{ancilla}_k}_k :=
  \frac{1}{\sqrt{3^{n - 1}}} \sum_{\vec{r}_{-k} \in \{0,1,2\}^{n-1}}
  \ket{\vec{r}_{-k}}$. Define the dilated state 
  \[\ket{\psi'} := \ket{\psi} \ot \anck{1}_1^A \ot \dots \ot
  \anck{n}_n^A \ot \anck{1}_1^B \ot \dots \ot
    \anck{n}_n^B \]
  and define dilated observables on Alice's side
  \begin{align} &\tilde{A}^{c}_{r, k} := \sum_{\vec{r}_{-k}} \sum_{\vec{a}_0,
                  \vec{a}_1} (-1)^{(\vec{a}_c)_k}
                  P^{\vec{a}_0 , \vec{a}_1}_{\vec r } \ot I_1 \ot \dots \ot I_{k -1} \ot
                  \proj{\vec{r}_{-k}} \ot I_{k+1} \dots \ot I_{n} \nonumber \\
                & = \sum_{\vec{r}_{-k}} A_{\vec s , \vec{e}_k}^{\vec c} \ot I_1 \ot \dots \ot I_{k -1} \ot
                  \proj{\vec{r}_{-k}} \ot I_{k+1} \dots \ot I_{n} \nonumber 
  \end{align}
  
  Where $\vec c$ in the last line can be any $\vec c$ satisfying $\vec c_k = c$, and wherever we write a sum over $\vec r_{-k}$ it is implicit that $r_k$ is fixed to be $r_k = r$.

  Observe that the operators $\tilde{A}^{c}_{r,k}$ are true
  observables, i.e. they are Hermitian and square to $I$. Moreover,
  $\tilde{A}^{c}_{r,k}$ simulates the two-outcome POVM whose
  elements are given by
  $M^{a_c} := \E_{\vec{r}_{-k}} P^{a_c}_{\vec{r}, k}$.
  
  Similarly, define dilated observables on Bob's side
    \begin{align} &\tilde{B}^{r}_{c, k} := \sum_{\vec{c}_{-k}} \sum_{\vec{b}_0,
                    \vec{b}_1} (-1)^{(\vec{b}_r)_k}
                    Q^{\vec{b}_0 , \vec{b}_1}_{\vec c } \ot I_1 \ot \dots \ot I_{k -1} \ot
                    \proj{\vec{r}_{-k}} \ot I_{k+1} \dots \ot I_{n} \nonumber \\
                  & = \sum_{\vec{c}_{-k}} B_{\vec c , \vec{e}_k}^{\vec r} \ot I_1 \ot \dots \ot I_{k -1} \ot
                    \proj{\vec{r}_{-k}} \ot I_{k+1} \dots \ot I_{n} \nonumber 
    \end{align}
    
 Where $\vec r$ in the last line can be any $\vec r$ satisfying $\vec r_k = r$, and wherever we write a sum over $\vec c_{-k}$ it is implicit that $c_k$ is fixed to be $c_k = c$.

    Observe that the operators $\tilde{B}^{r}_{c, k}$ are true
    observables, i.e. they are Hermitian and square to $I$. Moreover,
    $\tilde{B}^{r}_{c, k}$ simulates the two-outcome POVM whose
    elements are given by
    $M^{b_c} := \E_{\vec{c}_{-k}} P^{b_r}_{\vec{c}, k}$.

\end{definition}

\begin{lemma}
  \label{lem:single_round_phase}
  For all $k, r, r', c, c'$, it holds that \[ \|( \tA^{c}_{r, k} \tA^{c'}_{r', k}
  - (-1)^{f(r,r', c, c')} \tA^{c'}_{r', k} \tA^{c}_{r,k})\ket{\psi'} \|
  \leq O(\sqrt{\eps}).\]
  
  The analogous statement also holds for Bob operators.  \mnote{Should we leave this sentence like this?}
\end{lemma}
\begin{proof}
  Follows from single round analysis.  See Appendix \ref{app:singleroundcase}.  Replacing the operators $A^c_r$ in that analysis with $\tA^{c}_{r, k}$, and replacing $B^r_c$ in that analysis with $\tB^{r}_{c, k}$ one may observe that the analysis in  Appendix \ref{app:singleroundcase} still holds.    \mnote{Hopefully.  Still need to incorporate the average over all other rounds.  Should check --- note to myself}
\end{proof}
\begin{lemma}
  \label{lem:com_between_rounds}
  For all $k \neq k', r, r', c, c'$, it holds that 
  \[ \| (\tA^c_{r, k} \tA^{c'}_{r', k'}  -
  \tA^{c'}_{r', k'} \tA^{c}_{r,k} )\ket{\psi'}\| \leq O(\sqrt{\eps}).\]
  
  The analogous statement also holds for Bob operators.  \mnote{Should we leave this sentence like this?}
\end{lemma}
\begin{proof}
  Let $\vec{c}$ be any choice of columns such that $c_k = c, c_{k'} =
  c'$.

  Recall that by equation \eqref{eq:magicwin} we have that 
  
  \begin{equation}
    \forall \vec{p}, \quad \E_{\vec{r}, \vec{c}} \bra{\psi} A_{\vec r, \vec p}^{\vec c} B_{\vec
      c, \vec p}^{\vec r} \ket{\psi} \geq 1 - \eps.
  \end{equation}
  
  Setting $\vec p = \vec{e}_{k}$ gives that, for all fixed values of $r_{k}$ and $c_{k}$,
  
  \begin{equation}
    \forall k, \quad \E_{\vec{r}_{-k}, \vec{c}_{-k}} \bra{\psi} A_{\vec r, \vec{e}_{k}}^{\vec c} B_{\vec
      c, \vec{e}_{k}}^{\vec r} \ket{\psi} \geq 1 - 9 \eps.
  \end{equation}
  
  So,
  
  \begin{equation}
    \forall k, \quad \E_{\vec{r}_{-k}, \vec{c}_{-k}} d_\psi \left (A_{\vec r, \vec{e}_{k}}^{\vec c} B_{\vec
        c, \vec{e}_{k}}^{\vec r} \right )^2 \leq 18
    \eps 
    \label{eq:switchab}
  \end{equation}
  \anote{Factor of 2 could be absorbed into def of $d_\psi$.}
  
  
  Further, recall that $A^{\vec{c}}_{\vec{r},
    \vec{e}_k} = A^{\vec{c}'}_{\vec{r}, \vec{e}_k}$
  as long as the $k$th coordinate of $\vec{c}$ and
  $\vec{c}'$ agree. Denote by $\E_{\vec{c} | k,k'}$ the uniform
  distribution over choices of column vector $\vec{c}$ such that
  $c_k = c$ and $c_{k'} = c'$. Then
    \begin{align*}
    d_\psi(\tA^c_{r,k} \tA^{c'}_{r', k'},
    \tA^{c'}_{r'k'}, \tA^c_{r,k}) &= \E_{\vec{c} | k, k'}
                                    d_{\psi'}\big(\sum_{\vec{r}_{-k}} \sum_{\vec{r}'_{-k'}}
                                    A^{\vec{c}}_{\vec{r}, \vec{e}_k} 
                                    A^{\vec{c}}_{\vec{r}',
                                    \vec{e}_{k'}}
                                    \ot\proj{\vec{r}_{-k},
                                    \vec{r}'_{-k'}}_{k,k'}, \\
      &\qquad\qquad \sum_{\vec{r}_{-k}} \sum_{\vec{r}'_{-k'}}
                                    A^{\vec{c}}_{\vec{r}',
                                    \vec{e}_{k'}}
                                    A^{\vec{c}}_{\vec{r}, \vec{e}_k} 
                                    \ot\proj{\vec{r}_{-k},
                                    \vec{r}'_{-k'}}_{k,k'}\big) \\
      \intertext{Note that the column vector $\vec{c}$ is common to
      both $A$ operators. Also, as a convention, wherever there is a sum or expectation over $\vr_{-k}$ or $\vr'_{-k'}$ in this proof, it is implicit that the values of $r_k$ and $r'_{k'}$ are fixed to be $r_k = r$ and $r'_{k'} = r'$. Now, we apply \lemref{tri} to move
      the leftmost $A$ operator to Bob.}
      &\leq \E_{\vec{c} | k, k'} \Big[
                                    d_{\psi'}\big(\sum_{\vec{r}_{-k}}
        \sum_{\vec{r}'_{-k'}}
        A^{\vec{r}}_{\vec{c},
                                    \vec{e}_{k}}
                                    B^{\vec{r}'}_{\vec{c}, \vec{e}_{k'}} 
                                    \ot\proj{\vec{r}_{-k},
                                    \vec{r}'_{-k'}}_{k,k'}, \\
      &\qquad\qquad \sum_{\vec{r}_{-k}} \sum_{\vec{r}'_{-k'}}
                                    A^{\vec{c}}_{\vec{r}',
                                    \vec{e}_{k'}}
                                    A^{\vec{c}}_{\vec{r}, \vec{e}_k} 
                                    \ot\proj{\vec{r}_{-k},
                                    \vec{r}'_{-k'}}_{k,k'}\big)  + \\
        &\qquad\qquad d_{\psi'} \big(\sum_{\vec{r}_{-k}}A^{\vec{c}}_{\vec{r},
          \vec{e}_k} \ot \proj{\vr_{-k}} \ot I_{k'}
          \sum_{\vec{r}'_{-k'}}A^{\vc}_{\vr', \ve_{k'}} \ot I_{k} \ot
          \proj{\vr'_{-k'}}, \\
      &\qquad\qquad\qquad
          \sum_{\vec{r}_{-k}} A^{\vc}_{\vr, \ve_k} \ot \proj{\vr_{-k}} \ot I_{k'}\sum_{\vec{r}'_{-k'}}B^{\vr'}_{\vc,
        \ve_{k'}}\ot I_{k} \ot \proj{\vr'_{-k'}}\big) \Big] \\
      \intertext{Note that $\|\sum_{\vec{r}_{-k}}A^{\vec{c}}_{\vec{r},
          \vec{e}_k} \ot \proj{\vr_{-k}} \ot I_{k'}\| \leq 1$. Hence,
      applying \lemref{tri2} and \lemref{coherent_average}, we get}
      &\leq \E_{\vec{c} | k, k'} \Big[
                                    d_{\psi'}\big(\sum_{\vec{r}_{-k}}
        \sum_{\vec{r}'_{-k'}}
        A^{\vec{r}}_{\vec{c},
                                    \vec{e}_{k}}
                                    B^{\vec{r}'}_{\vec{c}, \vec{e}_{k'}} 
                                    \ot\proj{\vec{r}_{-k},
                                    \vec{r}'_{-k'}}_{k,k'}, \\
      &\qquad\qquad \sum_{\vec{r}_{-k}} \sum_{\vec{r}'_{-k'}}
                                    A^{\vec{c}}_{\vec{r}',
                                    \vec{e}_{k'}}
                                    A^{\vec{c}}_{\vec{r}, \vec{e}_k} 
                                    \ot\proj{\vec{r}_{-k},
                                    \vec{r}'_{-k'}}_{k,k'}\big)  + \\
        &\qquad\qquad \E_{\vec{r}'_{-k'}} d_\psi(A^{\vc}_{\vr', \ve_{k'}},
          B^{\vr'}_{\vc,
        \ve_{k'}}) \Big] \\
      \intertext{By performing the same steps on the other $A$
      operator, we obtain}
      &\leq \E_{\vec{c} | k, k'} \Big[
                                    d_{\psi'}\big(\sum_{\vec{r}_{-k}} \sum_{\vec{r}'_{-k'}}
                                    B^{\vec{r}'}_{\vec{c}, \vec{e}_{k'}} 
                                    B^{\vec{r}}_{\vec{c},
                                    \vec{e}_{k}}
                                    \ot\proj{\vec{r}_{-k},
                                    \vec{r}'_{-k'}}_{k,k'}, \\
      &\qquad\qquad \sum_{\vec{r}_{-k}} \sum_{\vec{r}'_{-k'}}
                                    A^{\vec{c}}_{\vec{r}',
                                    \vec{e}_{k'}}
                                    A^{\vec{c}}_{\vec{r}, \vec{e}_k} 
                                    \ot\proj{\vec{r}_{-k},
                                    \vec{r}'_{-k'}}_{k,k'}\big)  + \\
        &\qquad\qquad \E_{\vec{r}_{-k}} d_\psi(A^{\vec{c}}_{\vec{r}, \vec{e}_k}, B^{\vec{r}}_{\vec{c},
        \vec{e}_k}) + \E_{\vr'_{-k'}} d_\psi(A^{\vc}_{\vr', \ve_{k'}}, B^{\vr'}_{\vc,
        \ve_{k'}}) \Big] \\
      \intertext{Now the $B$ operators can be commuted exactly since
      they share the same input $\vc$.}
      &\leq \E_{\vec{c} | k, k'} \Big[
                                    d_{\psi'}\big(\sum_{\vec{r}_{-k}} \sum_{\vec{r}'_{-k'}}
                                    B^{\vec{r}}_{\vec{c}, \vec{e}_{k}} 
                                    B^{\vec{r}'}_{\vec{c},
                                    \vec{e}_{k'}}
                                    \ot\proj{\vec{r}_{-k},
                                    \vec{r}'_{-k'}}_{k,k'}, \\
      &\qquad\qquad \sum_{\vec{r}_{-k}} \sum_{\vec{r}'_{-k'}}
                                    A^{\vec{c}}_{\vec{r}',
                                    \vec{e}_{k'}}
                                    A^{\vec{c}}_{\vec{r}, \vec{e}_k} 
                                    \ot\proj{\vec{r}_{-k},
                                    \vec{r}'_{-k'}}_{k,k'}\big)  + \\
        &\qquad\qquad \E_{\vr_{-k}} d_\psi(A^{\vec{c}}_{\vec{r}, \vec{e}_k}, B^{\vec{r}}_{\vec{c},
        \vec{e}_k}) + \E_{\vr'_{-k'}} d_\psi(A^{\vc}_{\vr', \ve_{k'}}, B^{\vr'}_{\vc,
        \ve_{k'}}) \Big] \\
      \intertext{We move the $B$s back to Alice by reversing the
      previous steps, again using Lemmas~\ref{lem:tri},
      \ref{lem:tri2}, and \ref{lem:coherent_average}}
            &\leq \E_{\vec{c} | k, k'} \Big[
                                    d_{\psi'}\big(\sum_{\vec{r}_{-k}} \sum_{\vec{r}'_{-k'}}
                                    A^{\vec{c}}_{\vec{r}', \vec{e}_{k'}} 
                                    A^{\vec{c}}_{\vec{r},
                                    \vec{e}_{k}}
                                    \ot\proj{\vec{r}_{-k},
                                    \vec{r}'_{-k'}}_{k,k'}, \\
      &\qquad\qquad \sum_{\vec{r}_{-k}} \sum_{\vec{r}'_{-k'}}
                                    A^{\vec{c}}_{\vec{r}',
                                    \vec{e}_{k'}}
                                    A^{\vec{c}}_{\vec{r}, \vec{e}_k} 
                                    \ot\proj{\vec{r}_{-k},
                                    \vec{r}'_{-k'}}_{k,k'}\big)  + \\
        &\qquad\qquad 2\E_{\vr_{-k}}d_\psi(A^{\vec{c}}_{\vec{r}, \vec{e}_k}, B^{\vec{r}}_{\vec{c},
        \vec{e}_k}) + 2\E_{\vr'_{-k'}}d_\psi(A^{\vc}_{\vr', \ve_{k'}}, B^{\vr'}_{\vc,
        \ve_{k'}})  \Big] \\
      &= \E_{\vec{c} | k, k'} (2\E_{vr_{-k}}d_\psi(A^{\vec{c}}_{\vec{r}, \vec{e}_k}, B^{\vec{r}}_{\vec{c},
        \vec{e}_k}) + 2\E_{\vr'_{-k'}}d_\psi(A^{\vc}_{\vr', \ve_{k'}}, B^{\vr'}_{\vc,
        \ve_{k'}})) \\
      \intertext{Finally, we bound this by \myeqref{switchab}.  Note that \myeqref{switchab} is stated with $\E_{ \vec r_{-k}, \vec{c}_{-k}}$, but this implies the same statement with  $\E_{\vec{c} | k, k'}\E_{\vec r_{-k}}$ with an additional constant factor of 3.  Similarly for $\E_{\vec{c} | k, k'}\E_{\vec r_{-k'}}$.  So, continuing our computation: }
      &\leq 4 \cdot 3 \cdot 3 \sqrt{2 \eps} = 36\sqrt{2\eps}.
  \end{align*}
  \anote{Check square roots, also some of the invocations of
    \lemref{tri} should really be Cauchy-Schwartz.}
\end{proof}

\begin{lemma}
\label{lem:switchtatb}
\[\forall r,c,k, \quad d_{\psi'} \left (\tA^{c}_{r,
    k},\tilde{B}^{r}_{c, k}\right ) \leq  O(\sqrt{\eps})\] 
\end{lemma}
\begin{proof}
  In the argument below, let $\vec{r}$ be the row vectors agreeing
  with $r$ on index $k$ and $\vec{r}_{-k}$ on the remaining indices;
  likewise for $\vec{c}$ (note that $\vec{r}_{-k}$ is stored in
  Alice's register and and $\vec{c}_{-k}$ in Bob's). The main trick is
  to use the
  freedom of choice of $\vec{c}$ on Alice's operators to pick
  $\vec{c}$ agreeing with Bob's ancilla register $\vec{c}_{-k}$.
\begin{align*}
  d_{\psi'}(\tA^c_{r,k} , \tB^{r}_{c,k})^2 
  &= \| \frac{1}{3^{n-1}}\sum_{\vec{r}_{-k}, \vc_{-k}}
    A^{\vec{c}}_{\vr, \ve_k} \ket{\psi}_{AB}
    \ot \ket{\vec{r}_{-k}}^A_{k} \ot
    \ket{\vec{c}_{-k}}^B_k -  \\
  &\qquad\qquad \frac{1}{3^{n-1}}\sum_{\vec{r}_{-k}, \vc_{-k}}
    B^{\vec{r}}_{\vc, \ve_k} \ket{\psi}_{AB}
    \ot \ket{\vec{r}_{-k}}^A_{k} \ot
    \ket{\vec{c}_{-k}}^B_k\|^2 \\
  \intertext{By \lemref{coherent_average} with $i = (\vec{r}_{-k},
  \vc_{-k})$,}
  &= \E_{\vec{r}_{-k}, \vc_{-k}} d_{\psi'}(A^{\vc}_{\vr, k},
    B^{\vr}_{\vc, k})^2 \\
  \intertext{This is bounded by the probability that round $k$ of the
  test succeeds with inputs $r$ and $c$}
  &\leq O(\eps).
\end{align*}
\end{proof}

\begin{lemma} \label{lem:movetatotb2}
$\forall \vec r, \vec c, \vec p$ and $\forall i \in [n]$ 

$$\left  | \bra{\psi'} (\prod_{k=n}^{i+1} (\tB^{r_k}_{c_k, k})^{p_k}) A_{\vec r, \vec{p}}^{\vc}(\prod_{k=1}^i \tA^{c_k}_{r_k, k})\ket{\psi'}  - \bra{\psi'} (\prod_{k=n}^i (\tB^{r_k}_{c_k, k})^{p_k}) A_{\vec r, \vec{p}}^{\vc}(\prod_{k=1}^{i-1} \tA^{c_k}_{r_k, k})\ket{\psi'}\right | \leq O(\sqrt{\eps})  $$
\end{lemma}
\begin{proof}
 Fixing $\vec r, \vec c, \vec p$, and fixing  $i \in [n]$ we have
  
  \begin{align*}
  &\left  | \bra{\psi'} (\prod_{k=n}^{i+1} (\tB^{r_k}_{c_k, k})^{p_k}) A_{\vec r, \vec{p}}^{\vc}(\prod_{k=1}^i \tA^{c_k}_{r_k, k})\ket{\psi'}  - \bra{\psi'} (\prod_{k=n}^i (\tB^{r_k}_{c_k, k})^{p_k}) A_{\vec r, \vec{p}}^{\vc}(\prod_{k=1}^{i-1} \tA^{c_k}_{r_k, k})\ket{\psi'}\right | \\
  &  = \left  | \bra{\psi'} (\prod_{k=n}^{i+1} (\tB^{r_k}_{c_k, k})^{p_k}) A_{\vec r, \vec{p}}^{\vc}(\prod_{k=1}^{i-1} \tA^{c_k}_{r_k, k}) \tB^{r_i}_{c_i, i}\ket{\psi'}   +  \bra{\psi'} (\prod_{k=n}^{i+1} (\tB^{r_k}_{c_k, k})^{p_k}) A_{\vec r,\vec{p}}^{\vc}(\prod_{k=1}^{i-1} \tA^{c_k}_{r_k, k}) \left (\tA^{r_i}_{c_i,i} - \tB^{r_i}_{c_i,i} \right )\ket{\psi'} \right . \\
  &\qquad \left . - \bra{\psi'} (\prod_{k=n}^i (\tB^{r_k}_{c_k, k})^{p_k}) A_{\vec r, \vec{p}}^{\vc}(\prod_{k=1}^{i-1} \tA^{c_k}_{r_k, k})\ket{\psi'}\right | \\
  & \leq \left  | \bra{\psi'} (\prod_{k=n}^{i+1} (\tB^{r_k}_{c_k, k})^{p_k}) A_{\vec r, \vec{p}}^{\vc}(\prod_{k=1}^{i-1} \tA^{c_k}_{r_k, k}) \tB^{r_i}_{c_i, i}\ket{\psi'}  - \bra{\psi'} (\prod_{k=n}^i (\tB^{r_k}_{c_k, k})^{p_k}) A_{\vec r, \vec{p}}^{\vc}(\prod_{k=1}^{i-1} \tA^{c_k}_{r_k, k})\ket{\psi'}\right | \\
  & \qquad + \left | \bra{\psi'} (\prod_{k=n}^{i+1} (\tB^{r_k}_{c_k, k})^{p_k}) A_{\vec r, \vec{p}}^{\vc}(\prod_{k=1}^{i-1} \tA^{c_k}_{r_k, k}) \left (\tA^{r_i}_{c_i,i} - \tB^{r_i}_{c_i,i} \right )\ket{\psi'} \right | \\
  & \leq \left  | \bra{\psi'} (\prod_{k=n}^{i+1} (\tB^{r_k}_{c_k, k})^{p_k}) \tB^{r_i}_{c_i, i} A_{\vec r, \vec{p}}^{\vc}(\prod_{k=1}^{i-1} \tA^{c_k}_{r_k, k}) \tB^{r_i}_{c_i, i}\ket{\psi'}  - \bra{\psi'} (\prod_{k=n}^i (\tB^{r_k}_{c_k, k})^{p_k}) A_{\vec r, \vec{p}}^{\vc}(\prod_{k=1}^{i-1} \tA^{c_k}_{r_k, k})\ket{\psi'}\right | \\
  & \qquad +  d_{\psi'}(\tA^{c_i}_{r_i,i}, \tB^{r_i}_{c_i,i}) \\
  &\leq 0 + O(\sqrt{\eps}) = O(\sqrt{\eps}).
  \end{align*}
  Here the last inequality uses Lemma \ref{lem:switchtatb}, and the second to last inequality uses that $\tB^{r_i}_{c_i,i}$ commutes with all Alice operators, and that $(\prod_{k=n}^{i+1} (\tB^{r_k}_{c_k, k})^{p_k}) A_{\vec r, \vec{p}}^{\vc}(\prod_{k=1}^{i-1} \tA^{c_k}_{r_k, k})$ is a unitary, so that 
  
  $$\left | \bra{\psi'} (\prod_{k=n}^{i+1} (\tB^{r_k}_{c_k, k})^{p_k}) A_{\vec r, \vec{p}}^{\vc}(\prod_{k=1}^{i-1} \tA^{c_k}_{r_k, k}) \left (\tA^{r_i}_{c_i,i} - \tB^{r_i}_{c_i,i} \right )\ket{\psi'} \right | \leq \|\left  (\tA^{r_i}_{c_i,i} - \tB^{r_i}_{c_i,i} \right )\ket{\psi'}  \| = d_{\psi'}(\tA^{c_i}_{r_i,i}, \tB^{r_i}_{c_i,i}).$$
\end{proof}

\begin{lemma} \label{lem:longindstep}
\begin{align}
    &  \Big| \E_{\vec r } \Big(\bra{\psi'} (\prod_{k=n}^{i+2}
      (\tB^{r_k}_{c_k, k})^{p_k}) (\prod_{k=n}^{i+1} A^{\vec c}_{\vec
      r, p_k \cdot \vec{e}_k} \ot I)( \tA^{c_{i+1}}_{r_{i+1},
      i+1})\ket{\psi'}  \nonumber \\
  &\qquad\qquad - \bra{\psi'} (\prod_{k=n}^{i+1} (\tB^{r_k}_{c_k, k})^{p_k}) (\prod_{k=n}^i A^{\vec c}_{\vec r, p_k \cdot \vec{e}_k}\ot I)( \tA^{c_i}_{r_i, i})\ket{\psi'} \Big) \Big| \leq O(\sqrt{\eps}) \\
    \intertext{and} 
    & \E_{\vec r } \left (1  - \bra{\psi'}   A^{\vec c}_{\vec r, p_n \cdot \vec{e}_n}( \tA^{c_n}_{r_n, n})\ket{\psi'} \right ) \leq  O(\sqrt{\eps}) \label{lastcase} 
    \end{align} 
\end{lemma}

\begin{proof}
\begin{align*}
    &  \left | \E_{\vec r } \left (\bra{\psi'} (\prod_{k=n}^{i+2} \tB^{r_k}_{c_k, k}) (\prod_{k=n}^{i+1} A^{\vec c}_{\vec r, p_k \cdot \vec{e}_k} \ot I)( \tA^{c_{i+1}}_{r_{i+1}, i+1})\ket{\psi'}  - \bra{\psi'} (\prod_{k=n}^{i+1} \tB^{r_k}_{c_k, k}) (\prod_{k=n}^i A^{\vec c}_{\vec r, p_k \cdot \vec{e}_k}\ot I)( \tA^{c_i}_{r_i, i})\ket{\psi'} \right ) \right | \\
    &  =  \left | \E_{\vec r } \left (\bra{\psi'} (\prod_{k=n}^{i+2} \tB^{r_k}_{c_k, k}) (\prod_{k=n}^{i+1} A^{\vec c}_{\vec r, p_k \cdot \vec{e}_k} \ot I)( \tA^{c_{i+1}}_{r_{i+1}, i+1})\ket{\psi'} - \bra{\psi'} (\prod_{k=n}^{i+2} \tB^{r_k}_{c_k, k}) (\prod_{k=n}^{i+1} A^{\vec c}_{\vec r, p_k \cdot \vec{e}_k} \ot I)( \tB^{r_{i+1}}_{c_{i+1}, i+1})\ket{\psi'} \right .  \right . \\
    & \left . \left . + \bra{\psi'} (\prod_{k=n}^{i+2} \tB^{r_k}_{c_k, k}) (\prod_{k=n}^{i+1} A^{\vec c}_{\vec r, p_k \cdot \vec{e}_k} \ot I)( \tB^{r_{i+1}}_{c_{i+1}, i+1})\ket{\psi'}  - \bra{\psi'} (\prod_{k=n}^{i+1} \tB^{r_k}_{c_k, k}) (\prod_{k=n}^i A^{\vec c}_{\vec r, p_k \cdot \vec{e}_k}\ot I)( \tA^{c_i}_{r_i, i})\ket{\psi'} \right ) \right | \\
    & \leq   \left | \E_{\vec r } \left (\bra{\psi'} (\prod_{k=n}^{i+2} \tB^{r_k}_{c_k, k}) (\prod_{k=n}^{i+1} A^{\vec c}_{\vec r, p_k \cdot \vec{e}_k} \ot I)( \tA^{c_{i+1}}_{r_{i+1}, i+1})\ket{\psi'} - \bra{\psi'} (\prod_{k=n}^{i+2} \tB^{r_k}_{c_k, k}) (\prod_{k=n}^{i+1} A^{\vec c}_{\vec r, p_k \cdot \vec{e}_k} \ot I)( \tB^{r_{i+1}}_{c_{i+1}, i+1})\ket{\psi'} \right )  \right | \\
    & + \left | \E_{\vec r } \left (  \bra{\psi'} (\prod_{k=n}^{i+2} \tB^{r_k}_{c_k, k}) (\prod_{k=n}^{i+1} A^{\vec c}_{\vec r, p_k \cdot \vec{e}_k} \ot I)( \tB^{r_{i+1}}_{c_{i+1}, i+1})\ket{\psi'}  - \bra{\psi'} (\prod_{k=n}^{i+1} \tB^{r_k}_{c_k, k}) (\prod_{k=n}^i A^{\vec c}_{\vec r, p_k \cdot \vec{e}_k}\ot I)( \tA^{c_i}_{r_i, i})\ket{\psi'} \right ) \right | \\
    & \leq \left | \E_{\vec r } \left (  d_{\psi'}(\tA^{c_{i+1}}_{r_{i+1}, i+1}, \tB^{r_{i+1}}_{c_{i+1}, i+1})  \right )  \right | \\
    & +  \left | \E_{\vec r } \left (  \bra{\psi'} (\prod_{k=n}^{i+1} \tB^{r_k}_{c_k, k}) (\prod_{k=n}^{i+1} A^{\vec c}_{\vec r, p_k \cdot \vec{e}_k} \ot I)\ket{\psi'}  - \bra{\psi'} (\prod_{k=n}^{i+1} \tB^{r_k}_{c_k, k}) (\prod_{k=n}^{i+1} A^{\vec c}_{\vec r, p_k \cdot \vec{e}_k}\ot I) \cdot A^{\vec c}_{\vec r, p_i \cdot \vec{e}_i}\ot I \cdot ( \tA^{c_i}_{r_i, i})\ket{\psi'} \right ) \right | \\
    & \leq \left | \E_{\vec r } \left (  d_{\psi'}(\tA^{c_{i+1}}_{r_{i+1}, i+1}, \tB^{r_{i+1}}_{c_{i+1}, i+1})  \right )  \right | +  \left | \E_{\vec r } \left ( d_{\psi'}(I, (A^{\vec c}_{\vec r, p_i \cdot \vec{e}_i}\ot I) \cdot  \tA^{c_i}_{r_i, i})  \right ) \right | \\
    \end{align*}
   
Where the second to last inequality uses the fact that $\| \bra{\psi'} (\prod_{k=n}^{i+2} \tB^{r_k}_{c_k, k}) (\prod_{k=n}^{i+1} A^{\vec c}_{\vec r, p_k \cdot \vec{e}_k} \ot I) \| = 1$, and the third inequality uses that fact that $\|\bra{\psi'} (\prod_{k=n}^{i+1} \tB^{r_k}_{c_k, k}) (\prod_{k=n}^{i+1} A^{\vec c}_{\vec r, p_k \cdot \vec{e}_k}\ot I)\| = 1$.  Now, applying Lemma \ref{lem:switchtatb}, we have

\begin{align}
    &  \left | \E_{\vec r } \left (\bra{\psi'} (\prod_{k=n}^{i+2} \tB^{r_k}_{c_k, k}) (\prod_{k=n}^{i+1} A^{\vec c}_{\vec r, p_k \cdot \vec{e}_k} \ot I)( \tA^{c_{i+1}}_{r_{i+1}, i+1})\ket{\psi'}  - \bra{\psi'} (\prod_{k=n}^{i+1} \tB^{r_k}_{c_k, k}) (\prod_{k=n}^i A^{\vec c}_{\vec r, p_k \cdot \vec{e}_k}\ot I)( \tA^{c_i}_{r_i, i})\ket{\psi'} \right ) \right | \nonumber \\
    & \leq  \E_{\vec r } \left (  d_{\psi'}(\tA^{c_{i+1}}_{r_{i+1}, i+1}, \tB^{r_{i+1}}_{c_{i+1}, i+1})  \right )   +   \E_{\vec r } \left ( d_{\psi'}(I, (A^{\vec c}_{\vec r, p_i \cdot \vec{e}_i}\ot I) \cdot  \tA^{c_i}_{r_i, i})  \right ) \nonumber \\
    & \leq O(\sqrt{\eps}) +  \E_{\vec r } \left ( d_{\psi'}(I, (A^{\vec c}_{\vec r, p_i \cdot \vec{e}_i}\ot I) \cdot  \tA^{c_i}_{r_i, i})  \right ) \label{eq:jump1} \\
\end{align}

And we note that

    \begin{align*}   
    &  \E_{\vec r } \left ( d_{\psi'}(I, (A^{\vec c}_{\vec r, p_i \cdot \vec{e}_i}\ot I) \cdot  \tA^{c_i}_{r_i, i})^2  \right )  \\
    & =  \E_{\vec r } \left ( \| \ket{\psi'} -  (A^{\vec c}_{\vec r, p_i \cdot \vec{e}_i}\ot I) \cdot  \tA^{c_i}_{r_i, i} \ket{\psi'} \|^2 \right )  \\
    & = \E_{\vec r } \left ( 2 - 2 \bra{\psi'} (A^{\vec c}_{\vec r, p_i \cdot \vec{e}_i}\ot I) \cdot  \tA^{c_i}_{r_i, i} \ket{\psi'}  \right ) \\
    &  = \E_{\vec r } \left ( 2 - 2 \left (\bra{\psi}\ot \frac{1}{\sqrt{3^{n - 1}}} \sum_{\vec{r}_{-k} \in \{0,1,2\}^{n-1}}
                            \bra{\vec{r}_{-k}}\right ) (A^{\vec c}_{\vec r, p_i \cdot \vec{e}_i}\ot I) \cdot  \left (  \sum_{\vec{r'}:  r'_{i} = r_i} A_{\vec r' , \vec{e}_i}^{\vec c} \ot
                      \proj{\vec{r'}_{-i}} \right )  \times ...  \right . \\
     & \left . ... \times \left (\ket{\psi}\ot \frac{1}{\sqrt{3^{n - 1}}} \sum_{\vec{r}_{-k} \in \{0,1,2\}^{n-1}}
                        \ket{\vec{r}_{-k}}\right )  \right ) \\
     &  = \E_{\vec r } \left ( 2 - 2 \cdot  \frac{1}{3^{n - 1}}   \sum_{\vec{r'}:  r'_{i} = r_i} \bra{\psi} A^{\vec c}_{\vec r, p_i \cdot \vec{e}_i} \cdot A_{\vec r' , \vec{e}_i}^{\vec c} \ket{\psi} \cdot  \bra{\vec{r'}_{-i}}\proj{\vec{r'}_{-i}} \ket{\vec{r'}_{-i}}  \right ) \\
     &= \E_{\vec r } \left ( 2 - 2 \cdot  \E_{\vec{r'}:  r'_{i} = r_i}   \bra{\psi} A^{\vec c}_{\vec r, p_i \cdot \vec{e}_i} \cdot A_{\vec r' , \vec{e}_i}^{\vec c} \ket{\psi}  \right )  = 2 \left ( 1 -  \E_{\vec r, \vec{r'}:  r'_{i} = r_i} \bra{\psi} A^{\vec c}_{\vec r, p_i \cdot \vec{e}_i} \cdot A_{\vec r' , \vec{e}_i}^{\vec c} \ket{\psi} \right ) \\
     & \leq 2 \cdot 3 \cdot 36 \eps \numberthis \label{prejensen} 
    \end{align*}

Where the last inequality follows from Lemma \ref{lem:moveatoa}.   Furthermore, by Jensen's inequality it follows that:

\begin{align*}   
&  \E_{\vec r } \left ( d_{\psi'}(I, (A^{\vec c}_{\vec r, p_i \cdot \vec{e}_i}\ot I) \cdot  \tA^{c_i}_{r_i, i})  \right ) \leq \sqrt{\E_{\vec r } \left ( d_{\psi'}(I, (A^{\vec c}_{\vec r, p_i \cdot \vec{e}_i}\ot I) \cdot  \tA^{c_i}_{r_i, i})^2  \right )}  \leq  O(\sqrt{\eps})
\end{align*}

Now, resuming the calculation in equation \eqref{eq:jump1}, we have that  

\begin{align*}
    &  \left | \E_{\vec r } \left (\bra{\psi'} (\prod_{k=n}^{i+2} \tB^{r_k}_{c_k, k}) (\prod_{k=n}^{i+1} A^{\vec c}_{\vec r, p_k \cdot \vec{e}_k} \ot I)( \tA^{c_{i+1}}_{r_{i+1}, i+1})\ket{\psi'}  - \bra{\psi'} (\prod_{k=n}^{i+1} \tB^{r_k}_{c_k, k}) (\prod_{k=n}^i A^{\vec c}_{\vec r, p_k \cdot \vec{e}_k}\ot I)( \tA^{c_i}_{r_i, i})\ket{\psi'} \right ) \right | \nonumber \\
    & \leq O(\sqrt{\eps}) +  \E_{\vec r } \left ( d_{\psi'}(I, (A^{\vec c}_{\vec r, p_i \cdot \vec{e}_i}\ot I) \cdot  \tA^{c_i}_{r_i, i})  \right ) \leq O(\sqrt{\eps})
    \end{align*}

\mnote{We should try to get this with $\sqrt{\eps}$ instead of $\eps^{1/4}$, since the latter is worse than Andrea's result!!! }

Finally, note that, since Equation \ref{prejensen} is valid for every $i$, it follows by the same calculation, with $i = n$, that:

\[ \left |  \E_{\vec r } \left (1  - \bra{\psi'}   A^{\vec c}_{\vec r, p_n \cdot \vec{e}_n}( \tA^{c_n}_{r_n, n})\ket{\psi'} \right ) \right | \leq  O(\eps) \leq O(\sqrt{\eps}) \]

\end{proof}

\begin{lemma}\label{lem:consistency_alice}
\[ \forall \vec c, \vec p, \quad \E_{\vec r } d_{\psi'} \left(A^{\vec{c}}_{\vec{r}, \vec p} \ot I, \prod_{k=1}^n (\tA^{c_k}_{r_k, k})^{p_k} \right )^2 \leq O(n \sqrt{\eps})  \]

\mnote{Need to account for the value of $\vec{p}$ throughout this Lemma and all sub-lemmas that it cites!!!} 

The analogous statement also holds for Bob operators \mnote{Should we leave this sentence this way?}
\end{lemma}
\begin{proof}
For simplicity of notation, throughout this proof, we will denote $A^{\vec{c}}_{\vec{r}} \ot I$ simply by $A^{\vec{c}}_{\vec{r}}$.  Start by noting that we have the following exact property:
  \[ A_{\vec r, \vec p}^{\vc} A_{\vec r, \vec{p}'}^{\vc} = A_{\vec{r},
    \vec{p} + \vec{p'}}^{\vc}. \]
  As a consequence, we may decompose each observable $A_{\vec r,
    \vp}^{\vc}$ into a product of \emph{single-round} observables
    
     \[ A_{\vec r, \vec{p}}^{\vc} = A_{\vr, p_1}^{\vc} \dots A_{\vr,
        p_k}^{\vc}. \]
        
  So, fixing any value of $\vec c$ , and  $\vec p$, we have 
  
  \begin{align*}
  & \E_{\vec r } d_{\psi'} \left(A_{\vec r, \vec{p}}^{\vc} , \prod_{k=1}^n (\tA^{c_k}_{r_k, k})^{p_k} \right )^2 \\
  &= \E_{\vec r } \left ( \bra{\psi'}A_{\vec r, \vec{p}}^{\vc \dagger} A_{\vec r, \vec{p}}^{\vc} \ket{\psi'} + \bra{\psi'} (\prod_{k=1}^n (\tA^{c_k}_{r_k, k})^{p_k})^{\dagger}(\prod_{k=1}^n (\tA^{c_k}_{r_k, k})^{p_k}) \ket{\psi'}  - \bra{\psi'}A_{\vec r, \vec{p}}^{\vc \dagger}(\prod_{k=1}^n (\tA^{c_k}_{r_k, k})^{p_k})\ket{\psi'} \right . \\
  & \left . - \bra{\psi'} (\prod_{k=1}^n (\tA^{c_k}_{r_k, k})^{p_k})^{\dagger} A_{\vec r, \vec{p}}^{\vc} \ket{\psi'} \right ) = 2 \E_{\vec r } \left (1  - \bra{\psi'}A_{\vec r, \vec{p}}^{\vc}(\prod_{k=1}^n (\tA^{c_k}_{r_k, k})^{p_k})\ket{\psi'}  \right )
  \end{align*} 
  Where, in the second equality we are using the fact that $A_{\vec r, \vec{p}}^{\vc}$ is Hermitian to get that \[\bra{\psi'}A_{\vec r, \vec{p}}^{\vc}(\prod_{k=1}^n (\tA^{c_k}_{r_k, k})^{p_k})\ket{\psi'}  =   \bra{\psi'}A_{\vec r, \vec{p}}^{\vc \dagger}(\prod_{k=1}^n (\tA^{c_k}_{r_k, k})^{p_k})\ket{\psi'} = \bra{\psi'} (\prod_{k=1}^n (\tA^{c_k}_{r_k, k})^{p_k})^{\dagger} A_{\vec r, \vec{p}}^{\vc} \ket{\psi'}.\]  Continuing, we have
  \begin{align*}
    & \E_{\vec r } d_{\psi'} \left(A_{\vec r, \vec{p}}^{\vc} , \prod_{k=1}^n (\tA^{c_k}_{r_k, k})^{p_k} \right )^2 \\
    &\qquad = 2 \E_{\vec r } \left (1  - \bra{\psi'}A_{\vec r, \vec{p}}^{\vc}(\prod_{k=1}^n (\tA^{c_k}_{r_k, k})^{p_k})\ket{\psi'}  \right ) \\
    &\qquad= 2 \E_{\vec r } \left (1  - \bra{\psi'} (\prod_{k=n}^2 (\tB^{r_k}_{c_k, k})^{p_k}) A_{\vec r, \vec{p}}^{\vc}( \tA^{c_1}_{r_1, 1})\ket{\psi'} \right . \\
    &\qquad\qquad \left. - \sum_{i= n}^1 \left ( \bra{\psi'} (\prod_{k=n}^{i+1} (\tB^{r_k}_{c_k, k})^{p_k}) A_{\vec r, \vec{p}}^{\vc}(\prod_{k=1}^i (\tA^{c_k}_{r_k, k})^{p_k})\ket{\psi'}  - \bra{\psi'} (\prod_{k=n}^i (\tB^{r_k}_{c_k, k})^{p_k}) A_{\vec r, \vec{p}}^{\vc}(\prod_{k=1}^{i-1} (\tA^{c_k}_{r_k, k})^{p_k})\ket{\psi'} \right)  \right ) \\
    &\qquad \leq  2 \E_{\vec r } \left (1  - \bra{\psi'} (\prod_{k=n}^2 (\tB^{r_k}_{c_k, k})^{p_k}) A_{\vec r, \vec{p}}^{\vc}( \tA^{c_1}_{r_1, 1})\ket{\psi'} \right ) \\
    &\qquad\qquad + \sum_{i= n}^1 2 \E_{\vec r } \left( \left | \bra{\psi'} (\prod_{k=n}^{i+1} (\tB^{r_k}_{c_k, k})^{p_k} ) A_{\vec r, \vec{p}}^{\vc}(\prod_{k=1}^i (\tA^{c_k}_{r_k, k})^{p_k})\ket{\psi'}  - \bra{\psi'} (\prod_{k=n}^i (\tB^{r_k}_{c_k, k})^{p_k}) A_{\vec r, \vec{p}}^{\vc}(\prod_{k=1}^{i-1} (\tA^{c_k}_{r_k, k})^{p_k})\ket{\psi'} \right|  \right ).\\
    \intertext{We now apply Lemma \ref{lem:movetatotb2} inside the expectation:}
    &\qquad \leq  2 \E_{\vec r } \left (1  - \bra{\psi'} (\prod_{k=n}^2 (\tB^{r_k}_{c_k, k})^{p_k}) A_{\vec r, \vec{p}}^{\vc}( \tA^{c_1}_{r_1, 1})\ket{\psi'} \right ) +  \sum_{i= 1}^n 2 \cdot O(\sqrt{\eps}) \\
    & \qquad =   2 \E_{\vec r } \left (1  - \bra{\psi'} (\prod_{k=n}^2 (\tB^{r_k}_{c_k, k})^{p_k}) A_{\vec r, \vec{p}}^{\vc}( \tA^{c_1}_{r_1, 1})\ket{\psi'} \right ) +   O(n \sqrt{\eps}) \\
    &\qquad = 2  \E_{\vec r } \left (1  - \bra{\psi'} (\prod_{k=n}^2 (\tB^{r_k}_{c_k, k})^{p_k}) (\prod_{k=n}^1 A^{\vec c}_{\vec r, p_k \cdot \vec{e}_k})( \tA^{c_1}_{r_1, 1})\ket{\psi'} \right ) +   O(n \sqrt{\eps}) \\
    &\qquad \leq  2 \left |  \E_{\vec r } \left (1  - \bra{\psi'}   A^{\vec c}_{\vec r, p_n \cdot \vec{e}_n}( \tA^{c_n}_{r_n, n})\ket{\psi'} \right ) \right | + 2 \sum_{i= n-1}^1 \Big|  \E_{\vec r }
      \Big(\bra{\psi'} (\prod_{k=n}^{i+2} (\tB^{r_k}_{c_k, k})^{p_k})
      (\prod_{k=n}^{i+1} A^{\vec c}_{\vec r, p_k \cdot \vec{e}_k})(
      \tA^{c_{i+1}}_{r_{i+1}, i+1})\ket{\psi'} \\
    &\qquad\qquad - \bra{\psi'} (\prod_{k=n}^{i+1} \tB^{r_k}_{c_k, k}) (\prod_{k=n}^i A^{\vec c}_{\vec r, p_k \cdot \vec{e}_k})( \tA^{c_i}_{r_i, i})\ket{\psi'} \Big) \Big|  +   O(n \sqrt{\eps})   \\     
    & \leq 2 \cdot O( \sqrt{\eps}) + 2(n-1) O(\sqrt{\eps})  +   O(n \sqrt{\eps}) = O(n \sqrt{\eps}) 
    \end{align*}
       
       \mnote{Need to finish putting in the $()^{p_k}$ corrections here!}
       
  Where the last inequality follows by Lemma \ref{lem:longindstep}.
 
\end{proof}

\subsection{The Isometry}
\label{sec:isometry}
\begin{definition}
  Define the \emph{single round} ``approximate Pauli'' operators on
  Alice's space by:
  \begin{align*}
    X_{2k-1}&= \tA_{1,k}^1 \\
    X_{2k} &= \tA_{1,k}^0 \\
    Z_{2k-1} &= \tA_{0,k}^0 \\
    Z_{2k} &= \tA_{0,k}^1.
  \end{align*}
  Likewise define the single round approximate Pauli operators on
  Bob's space by
  \begin{align*}
    X^B_{2k-1}&= \tB_{1,k}^1 \\
    X^B_{2k} &= \tB_{1,k}^0 \\
    Z^B_{2k-1} &= \tB_{0,k}^0 \\
    Z^B_{2k} &= \tB_{0,k}^1.
  \end{align*}
\end{definition}
\begin{lemma}[Approximate single-round Pauli relations]
  Suppose Alice and Bob share an entangled strategy that wins with
  probability $1 - \eps$. Then the single-round Pauli operators as
  defined above satisfy the following relations:
  \begin{equation}
  \begin{aligned}
    &\forall i, \quad d_\psi(X_i, X^B_i) &\leq \sqrt{\eps} \\
    &\forall i, \quad d_\psi(Z_i, Z^B_i) &\leq \sqrt{\eps} \\
    &\forall i, \quad d_\psi(X_iZ_i , -Z_i X_i) &\leq \sqrt{\eps} \\
    &\forall i \neq j, \quad d_\psi(X_iX_j, X_j X_i) &\leq \sqrt{\eps} \\
    &\forall i \neq j, \quad d_\psi(Z_iZ_j, Z_jZ_i) &\leq \sqrt{\eps}.
  \end{aligned}
  \label{eq:approx_single_pauli}
  \end{equation}
\end{lemma}
\begin{proof}
  The consistency relations follow from \lemref{switchtatb}. The other
  relations come from \thmref{approx_phase}.
\end{proof}
We will now build up multi-round Paulis from products of these.
\begin{lemma}[Approximate Pauli relations]
  Suppose $X_i$, $Z_i$ are observables on Alice and $X^B_i, Z^B_i$ are
  observables on Bob indexed by $i \in [n]$ satisfying
  \myeqref{approx_single_pauli}.
  Let $X^{\va} := \prod_{i=1}^{n} X_i^{a_i}$ and $Z^{\vb} :=
  \prod_{i=1}^{n} Z_i^{b_i}$, and likewise let $(X^B)^{\va} :=
  \prod_{i=n}^{1} (X^B_i)^{a_i}$ and $(Z^B)^{\vb} := \prod_{i=n}^{1} (Z^B_i)^{b_i}$. Then
  \begin{align}
    \forall \va, \vb, \va', \vb', \quad d_\psi((X^{\va}Z^{\vb})(X^{\va'}Z^{\vb'}), (-1)^{\va' \cdot
    \vb} X^{\va + \va'}
    Z^{\vb + \vb'} ) &\leq O(n^2\sqrt{\eps})     \label{eq:approx_pauli_phase}\\
    \forall \va, \vb,  \quad d_\psi((X^{\va}Z^{\vb}), (Z^B)^{\vb} (X^B)^{\va})
                     &\leq
                       O(n\sqrt{\eps}) \label{eq:approx_pauli_con} .
  \end{align}
  \label{lem:approx_pauli}
\end{lemma}
\begin{proof}
  \myeqref{approx_pauli_con} is an immediate consequence of
  \lemref{con2}. We obtain \myeqref{approx_pauli_phase} in two
  steps. First, by \myeqref{switchmany} of \lemref{switchmany}, we
  have that 
  \[ d_\psi(X^{\va} Z^{\vb}, (-1)^{\va \cd \vb} Z^{\vb} X^{\va}) \leq O(n^2 \sqrt{\eps}). \]
  Further, by \myeqref{riffle} of \lemref{switchmany} we have that
  \begin{align*}
    d_\psi(X^\va X^{\va'}, X^{\va + \va'}) &\leq O(n^2\sqrt{\eps}) \\
    d_\psi(Z^\vb Z^{\vb'}, Z^{\vb + \vb'}) &\leq O(n^2\sqrt{\eps}). \\
  \end{align*}
  Hence, 
  \begin{align*}
    d_\psi(X^{\va} Z^{\vb} X^{\va'} Z^{\vb'}, (-1)^{\va' \cd \vb}
    X^{\va + \va'} Z^{\vb + \vb'}) &\leq d_\psi(X^{\va} Z^{\vb}
                                     X^{\va'} Z^{\vb'}, (Z^B)^{\vb'}
                                     X^{\va} Z^{\vb} X^{\va'}) \\
    &\qquad + d_\psi((Z^B)^{\vb'} X^\va Z^{\vb} X^{\va'}, (-1)^{\va'
      \cd \vb} (Z^B)^{\vb'}X^\va
      X^{\va'} Z^{\vb}) \\
    &\qquad + d_\psi((-1)^{\va' \cd \vb} (Z^B)^{\vb'} X^\va X^{\va'}
      Z^{\vb}, (-1)^{\va' \cd \vb} (Z^B)^{\vb'} (Z^B)^\vb X^{\va }
      X^{\va'}) \\ 
    &\qquad +d_\psi((-1)^{\va' \cd \vb} (Z^B)^{\vb'} (Z^B)^\vb X^{\va}
      X^{\va'},
      (-1)^{\va' \cd \vb} (Z^B)^{\vb'} (Z^B)^\vb X^{\va + \va'}) \\
    &\qquad + d_\psi((-1)^{\va' \cd \vb} (Z^B)^{\vb'} (Z^B)^\vb X^{\va + \va'},
      (-1)^{\va' \cd \vb} X^{\va + \va'} Z^\vb Z^{\vb'}) \\
    &\qquad + d_\psi((-1)^{\va' \cd \vb} X^{\va + \va'} Z^{\vb}
      Z^{\vb'},
      (-1)^{\va' \cd \vb} X^{\va + \va'} Z^{\vb + \vb'}) \\
    &\leq O(n^2\sqrt{\eps}).
  \end{align*}
  \anote{Add explanatory text.}
\end{proof}


\begin{proof}[Proof of \thmref{iso}]
  Let $\WA_{\va, \vb} := X^{\va} Z^{\vb}$
  and $\WB_{\va, vb} := (X^B)^{\va} (Z^B)^{\vb}$, and let $\HH$ be the
  provers' Hilbert space, together with the ancillas adjoined in \secref{single_round_observables}. Then we define the
  isometry $V: \HH \to \HH \ot \C^{2n} \ot \C^{2n} \ot \C^{2n} \ot \C^{2n}$ by
  \[ V(\ket{\psi}) = \frac{1}{2^{3n}}\sum_{\va,\vb,\vc}
  \sum_{\vd,\ve,\vf} (-1)^{\vb \cd (\va+\vc)} (-1)^{\ve\cd(\vd
    + \vf)} \WA_{\va,\vb} \ot \WB_{\vd,\ve} \ket{\psi}
  \ot \ket{\va+\vc, \vc} \ot \ket{\vd+\vf, \vf}. \]
  Here the second and the fourth register are the ``output register''
  of the isometry, and the third and fifth register are ``junk.'' This
  isometry was introduced by McKague~\cite{McK16}, and has an
  alternate description in terms of a circuit that ``swaps'' the input
  into the output register, which is initialized to be maximally
  entangled with the junk register.

  We now show the expectation value of any multi-qubit Pauli operator
  on the output of the isometry is close to the corresponding
  expectation value of approximate Paulis in the isometry input. In
  the equations below, $\ket{\phi} = V(\ket{\psi})$, the Paulis
  $\sigma_X^A, \sigma_Z^A$ act on output register 2, and $\sigma_X^B,
  \sigma_Z^B$ on output register 4.
  \begin{align*}
    \mathcal{P} &= \bra{\phi} \sigma_X^A(\vs) \sigma_Z^A(\vt)  \sigma_X^B(\vu)
                  \sigma_Z^B(\vv) \ket{\phi}\\
               &= \frac{1}{2^{6n}} \sum_{\va,\vb,\vc} \sum_{\va',
                 \vb', \vc'} \sum_{\vd, \ve,\vf} \sum_{\vd', \ve',\vf'} \Big(
                 \bra{\psi} \ot \bra{\va'+\vc', \vc'}
                 \ot\bra{\vd' + \vf', \vf'} \WAd_{\va', \vb'}
                 \ot \WBd_{\vd', \ve'} (-1)^{\vb' \cd (\va' + \vc')
                 + \ve' \cd (\vd' + \vf')}\\
               &\qquad\qquad \times \sigma_X^A(\vs) \sigma_Z^A(\vt)  \sigma_X^B(\vu)
                 \sigma_Z^B(\vv) (-1)^{\vb \cd (\va + \vc) + \ve \cd (\vd+\vf)} \WA_{\va,\vb} \ot \WB_{\vd,\ve} \ket{\psi} \ot
                 \ket{\va+\vc, \vc} \ot \ket{\vd+\vf, \vf} \Big)\\
               &= \frac{1}{2^{6n}} \sum_{\va,\vb,\vc} \sum_{\va',
                 \vb', \vc'} \sum_{\vd, \ve, \vf} \sum_{\vd' ,\ve', \vf'}
                 \Big(
                 \bra{\psi} \ot \bra{\va' + \vc', \vc'} \ot \bra{\vd' + \vf', \vf'} \WAd_{\va', \vb'}
                 \ot \WBd_{\vd', \ve'} (-1)^{\vb'\cd (\va' + \vc')}
                 (-1)^{\ve' \cd (\vd' + \vf')}\\
               &\qquad\qquad \times (-1)^{(\vb + \vt)(\va+\vc)} (-1)^{(\ve +
                 \vv)(\vd + \vf)}\WA_{\va,\vb} \ot \WB_{\vd,\ve} \ket{\psi} \ot \ket{\va+\vc + \vs, \vc} \ket{\vd + \vf +
                 \vu, \vf} \Big)\\
               &= \frac{1}{2^{6n}} \sum_{\va,\vb,\vb',\vc} \sum_{\vd,\ve,\ve',\vf} \Big(\bra{\psi}
                 \WAd_{\va + \vs,\vb'} \ot \WBd_{\vd+\vu,
                 \ve'}(-1)^{\vb' \cd (\va+\vs+\vc)}
                 (-1)^{\ve' \cd (\vd+\vu+\vf)} \\
               &\qquad\qquad \times (-1)^{(\vb+\vt)\cd (\va+\vc)}
                 (-1)^{(\ve+\vv) \cd (\vd+\vf)} \WA_{\va,\vb}
                 \ot \WB_{\vd,\ve} \ket{\psi} \Big). \\
    \intertext{Now we do the sum over $\vc$ and $\vf$ to force $\vb' = \vb+ \vt$ and
    $\ve' = \ve + \vv$:}
               &= \frac{1}{2^{4n}} \sum_{\va,\vb} \sum_{\vd, \ve}
                 \Big( (-1)^{(\vb+\vt) \cd \vs}
                 (-1)^{(\ve+\vv)\cd \vu} \bra{\psi}
                 \WAd_{\va+\vs,\vb+\vt} \WA_{\va,\vb} \ot  \WBd_{\vd+\vu, \ve+\vv} \WB_{\vd,\ve}
                 \ket{\psi} \Big). \\
    \intertext{ Finally, we apply Lemma~\ref{lem:approx_pauli} to
    merge the $\WA$ and $\WB$ operators, picking up an error of
    $O(n^2\sqrt{\eps})$ in the process.}
               &\approx_{O(n^2\sqrt{\eps})} \bra{\psi} \WA_{\vs,\vt} \WB_{\vu,\vv} \ket{\psi}.
  \end{align*}
\end{proof}
\begin{lemma}
  Let $M_n$ be the $4n$-qubit operator defined by
  \[ M_n = \left(\frac{1}{2} IIII + \frac{1}{18}( IXIX + XIXI + XXXX + ZIZI + IZIZ + ZZZZ + XZXZ +
  ZXZX + YYYY)\right)^{\ot n}.\ \]
  Then if a density matrix $\rho$ satisfies $\Tr[M_n \rho] \geq 1
  -\delta$, $\bra{\mathrm{EPR}}^{\ot 2n} \rho \ket{\mathrm{EPR}}^{\ot
    2n} \geq 1 - \frac{9}{4}\delta$.
  \label{lem:honest_measurement}
\end{lemma}
\begin{proof}
  Observe that the highest eigenvalue of $M_1$ is $1$, with unique
  eigenvector $\EPR^{\ot 2}$. Moreover all other eigenvalues of $M_1$
  have absolute value at most $5/9$. Hence, the highest eigenvalue of $M_n$ is also $1$ with
  the unique eigenvector is $\EPR^{\ot 2n}$, and all other eigenvalues
  have absolute value at most $5/9$. Hence
  \[ M_n \leq \pEPR^{\ot 2n} + \frac{5}{9} (I - \pEPR^{\ot 2n}).\]
  So
  \begin{align*}
    1 - \delta &\leq \Tr[M_n \rho] \\
    &\leq \frac{4}{9} \Tr[\rho \pEPR^{\ot 2n}] + \frac{5}{9} \\
    \frac{4}{9} - \delta &\leq \frac{4}{9} \Tr[\rho \pEPR^{\ot 2n}] \\
    1 - \frac{9}{4} \delta &\leq \Tr[\rho \pEPR^{\ot 2n}].
  \end{align*}
\end{proof}

\begin{lemma}
  For every single round operator $\tA^{c}_{r, k}$, let $X^{\va}
  Z^{\vb}$ be
  approximate Pauli operator formed by taking the row-$r$, column-$c$
  entry in the Magic Square (\figref{1roundideal}), and converting $X$
  and $Z$ on the first and second qubits to the approximate Paulis on
  qubits $2k -1$ and $2k$, respectively. Then
  \[ d_\psi(\tA^{c}_{r,k}, X^{\va}Z^{\vb}) \leq O(\sqrt{\eps}). \]
  Likewise, for Bob,
  \[ d_\psi(\tB^{r}_{c,k}, (X^B)^{\va} (Z^B)^{\vb}) \leq
    O(\sqrt{\eps}). \]
  \anote{Is there a square root? Check with analogous manipulations in
    previous section.}
  \label{lem:tatopauli_single}
\end{lemma}
\begin{proof}
  First consider Alice. Then the conclusion follows by definition of
  the approximate Paulis for $r \in  \{0,1\}$. When $r = 2$, use the
  fact that $d_\psi(\tA^{c}_{2,k}, \tB^{2}_{c,k}) \leq
  O(\sqrt{\eps})$. \anote{Cite}
  By definition, $\tB^{2}_{ck} = -\tB^{1}_{ck} \tB^0_{ck}$. Each of
  these two operators can be switched back to Alice, to yield
  \[ d_\psi(\tA^{c}_{2,k}, -\tA^c_{0k} \tA^c_{1k}) \leq
  O(\sqrt{\eps}). \]
  This establishes the result for single round operators. 
  For the Bob, we follow the same argument, interchanging the role of
  the row and column indices.
\end{proof}
\begin{lemma}
  For every product of single-round operators $\prod_{k=1}^{n}
  (\tA^{c_k}_{r_k, k})^{p_k}$, let $X^{\va} Z^{\vb}$ be the
  approximate Pauli operator formed by applying the procedure of
  \lemref{tatopauli_single} to each single-round operator. Then 
  \[ d_\psi(\prod_{k=1}^{n} (\tA^{c_k}_{r_k,k})^{p_k}, X^{\va} Z^{\vb}) \leq O(n\sqrt{\eps}). \]
  The analogous statement holds for $B$.
  \label{lem:tatopauli}
\end{lemma}
\begin{proof}
  This is a consequence of \lemref{tatopauli_single} and \lemref{consistency_string}.
\end{proof}
\begin{lemma}
  Suppose Alice and Bob win the test with probability $1 - \eps$. Then
  for the operator $M_n$ defined in \lemref{honest_measurement}. 
  $\bra{\phi} M_n \ket{\phi} \geq 1- O(n^2\sqrt{\eps})$,
  where $\ket{\phi} = V(\ket{\psi})$ is the output of the isometry in
  \thmref{iso} applied to Alice and Bob's shared state $\ket{\psi}$.
  \label{lem:parallel_ideal_test}
\end{lemma}
\begin{proof}
 Recall from Fact~\ref{fact:magicwin}, we know that
  \[ \forall \vp, \quad \E_{\vr, \vc} \bra{\psi} A^{\vc}_{\vr, \vp} B^{\vr}_{\vc,
    \vp} \ket{\psi} \geq 1 - \eps. \]
  By applying the consistency relations \myeqref{consistency_alice}
  and \myeqref{consistency_bob} guaranteed by \thmref{approx_phase},
  we obtain that
  \[ \forall \vp, \quad \E_{\vr, \vc} \bra{\psi} \prod_{k=1}^{n}
  (\tA^{c_k}_{r_k,k})^{p_k} \prod_{k=1}^{n} (\tB^{r_k}_{c_k,k})^{p_k} \ket{\psi} \geq
  1 - O(n\sqrt{\eps}). \]
  Now, by \lemref{tatopauli}, we can switch the $\tA$ and $\tB$
  operators to approximate Paulis:
  \[ \forall \vp, \quad \E_{\vr, \vc} \bra{\psi} (X^{\va}
  Z^{\vb})((X^B)^{\vc} (Z^B)^{\vd}) \ket{\psi} \geq 1 - O(n\sqrt{\eps}). \]
  Applying \thmref{iso}, we obtain that 
  \[ \forall \vp, \quad \bra{\phi} \E_{\vr, \vc} (\sigma_X^A(\va)
  \sigma_Z^A(\vb)\sigma_X^B(\vc) \sigma_Z^B(\vd)) \ket{\phi} \geq 1 -
  O(n^2 \sqrt{\eps}).\]
  In particular, taking an expectation over uniformly random choices
  of $\vec{p}$, we obtain that
  \[ \bra{\phi} \E_{\vr, \vc, \vp} (\sigma_X^A(\va)
  \sigma_Z^A(\vb)\sigma_X^B(\vc) \sigma_Z^B(\vd)) \ket{\phi} \geq 1 -
  O(n^2\sqrt{\eps}). \]
  It is not hard to see that $\E_{\vr, \vc, \vp} (\sigma_X^A(\va)
  \sigma_Z^A(\vb)\sigma_X^B(\vc) \sigma_Z^B(\vd))$ is precisely the
  operator $M_n$, corresponding to the magic square test performed on
  an unknown state $\ket{\phi}$ using the measurement operators of the
  ideal strategy.
\end{proof}

\section{Discussion and open questions}

The reader familiar with previous self-testing results may notice that
our \thmref{iso} gives a robustness bound on the \emph{expectation
  value} of operators without explicitly characterizing the state,
whereas previous works often state a bound on the 2-norm
$\|V(\ket{\psi}) - \ket{\psi'} \ot \ket{\text{junk}}\|$, where
$\ket{\psi'}$ is a fixed target state. While it is
possible to translate from one to the other by means of the
techniques in \lemref{parallel_ideal_test}, we think the guarantee on
expectation values is more natural in applications where one
does not want to test closeness to a fixed target state, but rather to
test whether the state satisfies a certain \emph{property} described
by a measurement operator.

Self-testing and rigidity have been very active areas of research in
recent years, and we believe that many more interesting questions
remain to be answered. One open question of interest is to reduce the question and answer length of the test
without sacrificing the error scaling. This is especially interesting
from the perspective of computational complexity, where self-testing results
have been used to show computational hardness for estimating the value
of non-local games~\cite{Ji15,NV15}. Rigidity has also been
applied to secure delegated computation and quantum key
distribution: in particular, the work of Reichardt, Unger, and
Vazirani \cite{RUV13} achieves these applications using a serial
(many-round) version of the CHSH test; it would be interesting to see
if their results could be improved using the Magic Square test. 

A further way to generalize our result would be to adapt it to test
states made up of qudits, with local dimension $d \neq 2$. As our
techniques relied heavily on the algebraic structure of the qubit
Pauli group, this may require significant technical advances. In fact, a
variant of the Magic Square game for which the ideal strategy consists
of ``generalized Paulis'' (i.e. the mod $d$ shift- and clock-matrices)
was recently proposed by McKague~\cite{McK16b}, and it would be
interesting to see if our analysis could extend to the parallel
repetition of this game. Likewise, it would be interesting to extend
our analysis to states other than the EPR state---for instance, could
we do something like McKague's self-test for $n$-qubit graph
states~\cite{McK16}, but with only two provers instead of $n$?

\section{Acknowledgements}
AN was supported by ARO contract W911NF-12-1-0486.  MC was supported by the National Science Foundation under Grant Number 0939370. Some of this
research was conducted while both authors were visiting the Perimeter
Institute. We thank Matthew McKague and
William Slofstra for helpful conversations, and Thomas Vidick for
helpful conversations and for inspiring the proof of \lemref{honest_measurement}. We also thank
Andrea Coladangelo for communicating with us regarding~\cite{Coladangelo16}.
\bibliography{magicsquare}

\pagebreak

\appendix

\section{Properties of the State-Dependent Distance}

\begin{definition}
  Given a state $\ket{\psi}$ and two operators $A, B$, the
  \emph{state-dependent distance} $d_{\psi}(A, B)$ between $A$ and $B$
  is defined to be
  \[ d_{\psi}(A, B) := \| A \ket{\psi} - B \ket{\psi} \|. \]
\end{definition}
\begin{lemma}
  \label{lem:tri}
  The state-dependent distance satisfies the triangle inequality
  \[ \forall A, B, C, \quad d_{\psi}(A, C) \leq d_{\psi}(A, B) +
  d_{\psi}(B, C). \]
\end{lemma}
\begin{lemma}
  \label{lem:tri2}
  Let $A, B, C, D$ be bounded operators. Then
  \[ d_{\psi}(DA, DC) \leq d_{\psi}(DA, DB) + \|D\| d_{\psi}(B, C). \]
\end{lemma}
\begin{proof}
  By \lemref{tri}, \[ d_\psi(DA, DC) \leq d_{\psi}(DA, DB) +
  d_{\psi}(DB, DC). \]
  Expand the second term:
  \begin{align*}
    d_{\psi}(DB, DC) &= \| D(B \ket{\psi} - C \ket{\psi}) \|_2 \\
    &\leq \|D \|  \cdot \|B\ket{\psi} - C\ket{\psi}\|_2 \\
    &= \|D\| d_{\psi}(B, C).
  \end{align*}
\end{proof}

The following lemma tells us that guarantees on the state-dependent
distance on average can be made ``coherent.''
\begin{lemma}
  Let $\{A_i\}$ and $\{B_i\}$ be two sets of operators indexed by $i \in
  [N]$, and suppose that \[\E_{i} d_\psi (A_i, B_i)^2 = \delta.\] Define
  the extended state $\ket{\psi'} = \frac{1}{\sqrt{N}} \sum_{i \in
    [N]} \ket{\psi} \ot \ket{i}$, and the extended operators
  $\tA = \sum_i A_i \ot \proj{i}$ and $\tB = \sum_i B_i \ot
  \proj{j}$. Then \[d_{\psi'} (\tA, \tB)^2 = \delta.\]
  \label{lem:coherent_average}
\end{lemma}
\begin{proof}
  \begin{align*}
    d_{\psi'}(\tA, \tB) &= \| \tA\ket{\psi'} - \tB\ket{\psi'}\|^2 \\
                        &= \| \frac{1}{\sqrt{N}}\sum_i A_i \ket{\psi} \ot \ket{i} - \frac{1}{\sqrt{N}}\sum_i B_i \ket{\psi}
                          \ot \ket{i} \|^2 \\
                        &= \frac{1}{N} \sum_i \bra{\psi} (A_i^\dagger A_i +
                          B_i^\dagger B_i - A_i^\dagger B_i - B_i^\dagger A_i) \ket{\psi}  
    \\
                        &= \E_{i} d_\psi(A_i, B_i)^2 \\
                        &= \delta.
  \end{align*}
\end{proof}

\begin{lemma}  \label{lem:maybesaveeps}
Given three Hermitian, unitary  operators $T, T', S$, and a unit vector $\ket{\sigma}$, if:
$\bra{\sigma}T \cdot S\ket{\sigma} \geq 1 - \delta $ and
$\bra{\sigma}T' \cdot S\ket{\sigma} \geq 1 - \delta,$ then
$\bra{\sigma}T \cdot T'\ket{\sigma} \geq 1 - 4 \delta$.
\end{lemma}
\begin{proof}
Note that 
\begin{align*}
& \| (T-S) \ket{\sigma} \|^2  = 2 - 2  \bra{\sigma}T \cdot S\ket{\sigma} \leq 2 \delta \\ 
\intertext{and, similarly,}
& \| (T'-S) \ket{\sigma} \|^2  = 2 - 2  \bra{\sigma}T' \cdot S\ket{\sigma} \leq 2 \delta. 
\end{align*}
So, by the Cauchy-Schwarz inequality,
\[ \left | \bra{\sigma}(T-S)(T'-S) \ket{\sigma} \right | \leq \| (T-S) \ket{\sigma} \| \cdot \| (T'-S) \ket{\sigma} \| \leq \sqrt{2 \delta} \cdot \sqrt{2 \delta} = 2 \delta.  \]
Expanding out the Left Hand Side, now gives
\begin{align*}
&2 \delta \geq \left | \bra{\sigma}(T-S)(T'-S) \ket{\sigma} \right |  = \left | \bra{\sigma}T \cdot T'\ket{\sigma} - \bra{\sigma}T \cdot S \ket{\sigma} - \bra{\sigma}S \cdot T'\ket{\sigma} + \bra{\sigma}S \cdot S \ket{\sigma} \right | \\
& =   \left | \bra{\sigma}T \cdot T'\ket{\sigma} - \bra{\sigma}T \cdot S \ket{\sigma} - \bra{\sigma}S \cdot T'\ket{\sigma} + 1 \right | \\
\intertext{So,} 
& -2 \delta \leq \bra{\sigma}T \cdot T'\ket{\sigma} - \bra{\sigma}T \cdot S \ket{\sigma} - \bra{\sigma}S \cdot T'\ket{\sigma} + 1 \\
\intertext{and} 
&  \bra{\sigma}T \cdot T'\ket{\sigma} \geq \bra{\sigma}T \cdot S \ket{\sigma} + \bra{\sigma}S \cdot T'\ket{\sigma} - 1 - 2 \delta  \geq (1- \delta) + (1- \delta) - 1 - 2 \delta  = 1 - 4 \delta, \\
\end{align*}
where the last inequality again uses the assumption of this lemma.
\end{proof}

We now state and prove some ``utility''
lemmas, about what happens when we commute words of operators past
each other.
\begin{lemma}
  Let $A_1, \dots, A_k$ be Hermitian operators on Alice's space, and $B_1,
  \dots, B_k$ be
  Hermitian operators on Bob's space, such that
  \[ \forall i, \quad d_\psi(A_i, B_i) \leq \eps_i. \]
  Then \[ d_\psi(\prod_{i=1}^k A_i ,\prod_{i=k}^1 B_i) \leq
  \sum_{i=1}^k \eps_i \]
  \label{lem:con2}
\end{lemma}
\begin{proof}
  \begin{align*}
    d_\psi(\prod_{i=1}^k A_i,\prod_{i=k}^1 B_i) &\leq d_\psi(A_1 \dots
                                                  A_k, B_k A_1 \dots
                                                  A_{k-1}) +
                                                  d_\psi(B_k A_1 \dots
                                                  A_{k-1}, B_k B_{k-1}
                                                  A_1 \dots A_{k-2})  \\
                                                &\qquad +
                                                  \dots + d_\psi(B_k
                                                  \dots B_{2} A_1, B_k
                                                  \dots B_1) \\
    &\leq d_\psi(A_k, B_k) + d_\psi(A_{k-1}, B_{k-1}) + \dots +
      d_\psi(A_1, B_1) \\
    &= \sum_i \eps_i
  \end{align*}
  \anote{Add some text}
\end{proof}
\begin{lemma}
  Let $A_1, \dots A_k$ and $A'_1, \dots A'_k$ be operators on Alice,
  and $B_1, \dots B_k$ be operators on Bob, such that
  \begin{align*}
    \forall i, \quad d_\psi(A_i, B_i) &\leq \eps_1 \\
    \forall i, \quad d_\psi(A'_i, B_i) &\leq \eps_2.
  \end{align*}
  Then
  \[ d_\psi(A_1 \dots A_k, A'_1 \dots A'_k) \leq n(\eps_1 +
  \eps_2). \]
  \label{lem:consistency_string}
\end{lemma}
\begin{proof}
  This is a straightforward application of the \lemref{con2}.
  \begin{align*}
    d_\psi(A_1 \dots A_k, A'_1 \dots A'_k) &\leq d_\psi(A_1 \dots A_k,
    B_k \dots B_1) + d_\psi(B_k \dots B_1, A'_1 \dots A'_k) \\
    &\leq n\eps_1 + n\eps_2.
  \end{align*}
\end{proof}

\begin{lemma}
  Let $A_1, \dots A_k$ be Hermitian operators on Alice's space, and $B_1, \dots,
  B_{k}$ be Hermitian operators on Bob's space. Suppose that
  \[ \forall i, \quad d_\psi(A_i, B_i) \leq \eps_1 \]
  and \[ \forall i, j \in \{1, \dots, k-1\}, j \in \{k\}, \quad d_\psi(A_i A_j,
  \alpha_{ij} A_j A_i) \leq \eps_2 \]
  where $\alpha_{ij} \in \{\pm 1\}$ for each choice of $i,j$. Then 
  \[ d_\psi(A_1 \dots A_k, \alpha_{1k} \alpha_{2k} \dots \alpha_{k-1,k}
  A_k  A_1 A_2\dots A_{k-1} ) \leq 2(k-2)\eps_1 + (k-1)\eps_2. \]
  \label{lem:switch3}
\end{lemma}
\begin{proof}
  \begin{align*}
    &d_\psi(A_1\dots A_k, (\prod_{i=1}^{k-1} \alpha_{ik}) A_k A_1 \dots
    A_{k-1}) \\
    &\qquad \leq d_\psi(A_1 \dots A_k, 
               \alpha_{k-1,k} A_1 \dots A_{k-2} A_k A_{k-1}) \\
    &\qquad\qquad + d_\psi(\alpha_{k-1,k} A_1 \dots A_{k-2} A_k
      A_{k-1},
      \alpha_{k-1,k} B_{k-1} A_1 \dots A_{k-2} A_k) \\
    &\qquad \qquad + d_\psi(\alpha_{k-1,k} B_{k-1} A_1 \dots A_{k-2}
      A_k,
      \alpha_{k-1,k} \alpha_{k-2,k} B_{k-1} A_1 \dots A_{k-3}
      A_k A_{k-2}) \\
    &\qquad\qquad+d_\psi(     \alpha_{k-1,k} \alpha_{k-2,k} B_{k-1} A_1 \dots A_{k-3}
      A_k A_{k-2},
      \alpha_{k-1,k} \alpha_{k-2,k} B_{k-1} B_{k-2} A_1 \dots A_{k-3}
      A_k) \\
    &\qquad \qquad + \dots \\
    &\qquad \qquad + d_\psi(\prod_{i=2}^{k-1} \alpha_{ik} B_{k-1} \dots B_{2}
      A_1 A_k, \prod_{i=1}^{k-1} \alpha_{ik} B_{k-1} \dots B_{2} A_k
      A_1) \\
    &\qquad \qquad + d_\psi(\prod_{i=1}^{k-1} \alpha_{ik} B_{k-1} \dots B_{2} A_k
      A_1, \prod_{i=1}^{k-1} \alpha_{ik} A_k A_1 \dots A_{k-1}) \\
    &\qquad \leq d_\psi(A_{k-1} A_{k}, \alpha{k-1,k} A_{k} A_{k-1}) +
      d_\psi(A_{k-1}, B_{k-1}) + \dots + d_\psi(A_2 A_{k}, \alpha{2k}
      A_{k} A_{2}) + d_\psi(A_2, B_2) \\
    &\qquad\qquad + d_\psi(A_1 A_k, \alpha{1k} A_k A_1) + d_\psi(B_2,
      A_2) + \dots + d_\psi(B_k, A_k) \\
    &\qquad\leq 2(k - 2)\eps_1 + (k-1)\eps_2
  \end{align*}
  \anote{Add some text}
\end{proof}

As a consequence of the preceding lemma

\begin{lemma}
  Let $S_1, \dots, S_k, T_1, \dots, T_k$ be Hermitian operators on Alice's space
  and let $S^B_1, \dots S^B_k, T_1^B \dots T^B_k$ be Hermitian operators on
  Bob's space, satisfying
  \begin{align*}
    \forall i, \quad d_\psi(S_i, S^B_i) &\leq \eps_1 \\
    \forall i, \quad d_\psi(T_i, T^B_i) &\leq \eps_2 \\
    \forall i,j, \quad d_\psi(S_i T_j, \alpha_{ij} T_j S_i) &\leq \eps_3.
  \end{align*}
  Then 
  \begin{equation}
    d_\psi(S_1 \dots S_k T_1 \dots T_k, \prod_{i,j=1}^{k} \alpha_{ij}
  T_1 \dots T_k S_1 \dots S_k) \leq 2(k-1) \eps_2 + k(2(k-1)\eps_1+ k\eps_3). 
  \label{eq:switchmany}
  \end{equation}
  Likewise,
  \begin{equation}
    d_\psi(S_1 \dots S_k T_1 \dots T_k, \prod_{i=2}^{k}
  \prod_{j=1}^{i-1} \alpha_{ij} S_1 T_1 S_2 T_2 \dots S_k T_k) \leq
  2(k-1)\eps_2 + \sum_{j=2}^{k} (2(j-2)\eps_2 + (j-1)\eps_3)
  \label{eq:riffle}
\end{equation}
  \label{lem:switchmany}
\end{lemma}
\begin{proof}
  We first prove \myeqref{switchmany}.
  \begin{align*}
    &d_\psi(S_1 \dots S_k T_1 \dots T_k, \prod_{i,j=1}^{k} \alpha_{ij}
    T_1 \dots T_k S_1 \dots S_k) \\
    &\qquad\leq d_\psi(S_1 \dots S_k T_1 \dots T_k, T^B_k \dots
      T^B_{2} S_1 \dots S_k T_1) \\
    &\qquad\qquad + d_\psi(T^B_k\dots T^B_2 S_1 \dots S_k T_1, 
      \prod_{i=1}^k \alpha_{i1} T^B_k\dots T^B_2 T_1 S_1 \dots S_k) \\
    &\qquad\qquad + d_\psi(\prod_{i=1}^k \alpha_{i1} T^B_k\dots T^B_2 T_1 S_1 \dots S_k, 
      \prod_{i=1}^k \alpha_{i1} T^B_k \dots T^B_3 T_1 S_1 \dots S_k
      T_2) \\
    &\qquad\qquad + d_\psi(\prod_{i=1}^k \alpha_{i1} T^B_k \dots T^B_3
      T_1 S_1 \dots S_k T_2,
      \prod_{i=1}^k \alpha_{i1} \alpha_{i2} T^B_k \dots T^B_3 T_1 T_2
      S_1 \dots S_k) \\
    &\qquad\qquad + \dots \\
    &\qquad\qquad + d_\psi(\prod_{i=1}^k \prod_{j=1}^{k-1} \alpha_{ij}
      T^B_{k} T_1
      \dots T_{k-1} S_1 \dots S_k,
      \prod_{i=1}^k \prod_{j=1}^{k-1} \alpha_{ij} T_1 \dots T_{k-1}
      S_1 \dots S_k T_k) \\
    &\qquad\qquad + d_\psi(\prod_{i=1}^k \prod_{j=1}^{k-1} \alpha_{ij}
      T_1 \dots T_{k-1} S_1 \dots S_k T_k, 
      \prod_{i,j=1}^{k} \alpha_{ij} T_1 \dots T_k S_1 \dots S_k) \\
    &\qquad\leq 2(k-1) \eps_2 + k(2(k-1)\eps_1+ k\eps_3). 
  \end{align*}
  \anote{Add some text.}
  The derivation of \myeqref{riffle} is very similar. The only
  difference is that the number of commutations of $S$ with $T$ is different.
\end{proof}

\section{The Single Round Case} \label{app:singleroundcase}
In this section, we review the self-testing result of~\cite{WBMS16} on
the single-round magic square game, and write out the measurement definitions concretely for use in our setting. The rules of the game are
described in Fig.~\ref{fig:magic}.
Any entangled strategy for this game is described by a shared quantum
state $\ket{\psi}_{AB}$ and projectors $P_{r}^{a_0, a_1}$ for Alice
and $Q_{c}^{b_0, b_1}$ for Bob. It can be seen that the game can be
won with certainty for the following strategy:
\begin{align*}
  \ket{\psi} &= \frac{1}{2} \sum_{i, j \in \{0, 1\}} \ket{ij}_A \ot
               \ket{ij}_B \\
  P_{0}^{a_0, a_1} &= \frac{1}{4}(I + (-1)^{a_0} Z)_{A1}
                     \ot (I + (-1)^{a_1} Z)_{A2} \ot I_{B} \\
  P_{1}^{a_0, a_1} &= \frac{1}{4}  (I + (-1)^{a_1} X)_{A1} \ot (I +
                     (-1)^{a_0} X)_{A2}
                     \ot I_{B} \\
  Q_{0}^{b_0, b_1} &= \frac{1}{4}I_A \ot (I + (-1)^{b_0} Z)_{B1} \ot
                     (I + (-1)^{b_1} X)_{B2} \\
  Q_{1}^{b_0, b_1} &= \frac{1}{4} I_A \ot (I + (-1)^{b_1} X)_{B1} \ot
                     (I + (-1)^{b_0} Z)_{B2} 
\end{align*}
This strategy is represented pictorially in
Fig.~\ref{fig:1roundideal}, where each row contains a set of
simultaneously-measurable observables that give Alice's answers, and
likewise each column for Bob.


Inspired by this ideal strategy, for \emph{any} strategy we can define
the following induced observables on Alice's system:
\begin{align*}
  X_1 &= \sum_{a_0, a_1} (-1)^{a_1} P_1^{a_0, a_1}  &= A^1_1\\
  X_2 &= \sum_{a_0, a_1} (-1)^{a_0} P_1^{a_0, a_1}  &= A^0_1\\
  Z_1 &= \sum_{a_0, a_1} (-1)^{a_0} P_0^{a_0, a_1} &= A^0_0\\
  Z_2 &= \sum_{a_0, a_1} (-1)^{a_1} P_0^{a_0, a_1} &= A^1_0,
\end{align*}
and on Bob's system:
\begin{align*}
  X_3 &= \sum_{b_0, b_1} (-1)^{b_1} Q_1^{b_0, b_1} &= B^1_1\\
  X_4 &= \sum_{b_0 b_1} (-1)^{b_1} Q_0^{b_0, b_1} &= B^1_0\\
  Z_3 &= \sum_{b_0, b_1} (-1)^{b_0} Q_0^{b_0, b_1} &= B^0_0\\
  Z_4 &= \sum_{b_0, b_1} (-1)^{b_0} Q_1^{b_0 b_1} &= B^0_1.
\end{align*}
The $X$ and $Z$ observables correspond to the first two rows and
columns of the square. From the third row and third column, we obtain
four more observables; two for Alice:
\begin{align*}
  W_1 &= \sum_{a_0, a_1} (-1)^{a_0} P_2^{a_0, a_1} &= A^0_2\\
  W_2 &= \sum_{a_0, a_1} (-1)^{a_1} P_2^{a_0,a_1} &= A^1_2,
\end{align*}
and two for Bob:
\begin{align*}
  W_3 &= \sum_{b_0, b_1} (-1)^{b_0} Q_2^{b_0, b_1} &= B^0_2\\
  W_4 &= \sum_{b_0, b_1} (-1)^{b_1} Q_2^{b_0, b_1} &= B^1_2.
\end{align*}

There are nine consistency conditions implied by winning the game with
probability $1 - \eps$:
\begin{align}
  \bra{\psi} Z_1 Z_3 \ket{\psi} &\geq 1 - 9\eps  \label{eq:win1}\\
  \bra{\psi} Z_2 Z_4 \ket{\psi} &\geq 1 - 9\eps \label{eq:win2}\\
  \bra{\psi} Z_1Z_2 W_3 \ket{\psi} &\geq 1 - 9\eps \label{eq:win3}\\
  \bra{\psi} X_2 X_4 \ket{\psi} &\geq 1- 9\eps \label{eq:win4}\\
  \bra{\psi} X_1 X_3 \ket{\psi} &\geq 1 - 9\eps \label{eq:win5}\\
  \bra{\psi} X_1X_2 W_4 \ket{\psi} &\geq 1 - 9\eps \label{eq:win6}\\
  -\bra{\psi} W_1 Z_3 X_4 \ket{\psi} &\geq 1 - 9\eps \label{eq:win7}\\
  -\bra{\psi} W_2 Z_4 X_3 \ket{\psi} &\geq 1 - 9\eps \label{eq:win8}\\
  -\bra{\psi} W_1 W_2 W_3 W_4 \ket{\psi} &\geq 1 - 9\eps. \label{eq:win9}
\end{align}

From this we obtain anticommutation conditions
\begin{align*}
  X_1 Z_1 &\approx X_1 Z_2 W_3  && \text{(by \eqref{eq:win3})} \\
          &= W_3 X_1 Z_2 \\
          &\approx W_3 X_1 Z_4  && \text{(by \eqref{eq:win2})} \\
          &\approx W_3 Z_4 X_3 && \text{(by \eqref{eq:win5})} \\
          &\approx - W_3W_2 && \text{(by \eqref{eq:win8})} \\
          &\approx W_3W_1W_3W_4 &&\text{(by \eqref{eq:win9})}\\
          &= W_1W_4 \\
          &\approx  - W_4 Z_3 X_4 &&\text{(by \eqref{eq:win7})}\\
          &\approx - Z_1 W_4 X_4 &&\text{(by \eqref{eq:win1})}\\
          &\approx -Z_1 X_2 W_4 &&\text{(by \eqref{eq:win4})}\\
          &\approx  -Z_1 X_2 X_2X_1 &&\text{(by \eqref{eq:win6})}\\
          &= -Z_1 X_1.
\end{align*}
We can also get commutation relations on different qubits:
\begin{align*}
  X_1 Z_2 &\approx X_1 Z_4 &&\text{(by \eqref{eq:win2})}\\
          &\approx Z_4 X_3 &&\text{(by \eqref{eq:win5})}\\
          &= X_3 Z_4 &&\text{(by construction)} \\
          &\approx X_3 Z_2 &&\text{(by \eqref{eq:win2})}\\
          &\approx Z_2 X_1 &&\text{(by \eqref{eq:win5})}.
\end{align*}

The other cases follow similarly.  See \cite{WBMS16} for further details.

\end{document}